\documentclass[a4paper,11pt,aps,pra,notitlepage,superscriptaddress,tightenlines,footinbib]{revtex4-1}

\usepackage[utf8]{inputenc}

\usepackage[a4paper]{geometry}

\usepackage{amsfonts}
\usepackage{amsmath}
\usepackage{amssymb}
\usepackage{amsthm}  
\usepackage{mathrsfs}
\usepackage{bbm}

\usepackage[colorlinks=true,linkcolor=black,citecolor=black,plainpages=false,pdfpagelabels]{hyperref}
\hypersetup{pdftitle={Template}}
\usepackage{mathtools}
\usepackage{color}

\theoremstyle{plain}
\newtheorem{thm}{Theorem}

\newtheorem{lem}[thm]{Lemma}
\newtheorem{prop}[thm]{Proposition}

\theoremstyle{remark}

\theoremstyle{definition}
\newtheorem{defi}{Definition}

\newcommand{\sket}[1]{{\ensuremath{\lvert#1\rangle}}}
\newcommand{\lket}[1]{{\ensuremath{\left\lvert#1\right\rangle}}}
\newcommand{\ket}[1]{\mathchoice{\lket{#1}}{\sket{#1}}{\sket{#1}}{\sket{#1}}}

\newcommand{\sbra}[1]{{\ensuremath{\langle#1\rvert}}}
\newcommand{\lbra}[1]{{\ensuremath{\left\langle#1\right\rvert}}}
\newcommand{\bra}[1]{\mathchoice{\lbra{#1}}{\sbra{#1}}{\sbra{#1}}{\sbra{#1}}}

\newcommand{\ident}{\id}
\DeclareMathOperator{\tr}{\mathrm{Tr}}

\newcommand{\demi}{\frac{1}{2}}

\newcommand{\C}{\mathbb{C}}

\newcommand{\HH}{\mathcal{H}}

\newcommand{\cS}{\mathcal{S}}
\newcommand{\cE}{\mathcal{E}}

\renewcommand{\epsilon}{\ensuremath\varepsilon}

\renewcommand{\phi}{\ensuremath{\varphi}}

\DeclareMathOperator{\Tr}{Tr}

\DeclareMathOperator{\T}{T}
\DeclareMathOperator{\supp}{supp}

\DeclareMathOperator{\id}{id}

\newcommand{\yb}{Y^{\frac{1-\alpha}{2\alpha}}}

\newcommand{\aaa}{\frac{1-\alpha}{2\alpha}}
\newcommand{\onehalf}{\frac{1}{2}}

\newcommand{\Dnew}{\widetilde{D}}
\newcommand{\Hnew}{\widetilde{H}}
\newcommand{\Dold}{D}
\newcommand{\Hold}{H}

\begin{document}

\title{\LARGE On quantum R\'enyi entropies: a new generalization and some properties}

\author{Martin M\"uller-Lennert}
\affiliation{Department of Mathematics, ETH Zurich, 8092 Z\"urich, Switzerland}
\author{Fr\'ed\'eric Dupuis}
\affiliation{Department of Computer Science, Aarhus University, 8200 Aarhus, Denmark}
\author{Oleg~Szehr}
\affiliation{Department of Mathematics, Technische Universität M\"unchen, 85748 Garching, Germany}
\author{Serge Fehr}
\affiliation{CWI (Centrum Wiskunde \& Informatica), 1090 Amsterdam, The Netherlands}
\author{Marco Tomamichel}
\affiliation{Centre for Quantum Technologies, National University of Singapore, Singapore 117543, Singapore}


\begin{abstract}
The R\'enyi entropies constitute a family of information measures that generalizes the well-known Shannon entropy, inheriting many of its properties. They appear in the form of unconditional and conditional entropies, relative entropies or mutual information, and have found many applications in information theory and beyond. Various generalizations of R\'enyi entropies to the quantum setting have been proposed, most prominently Petz's quasi-entropies and Renner's conditional min-, max- and collision entropy. However, these quantum extensions are incompatible and thus unsatisfactory.
We propose a new quantum generalization of the family of R\'enyi entropies that contains the von Neumann entropy, min-entropy, collision entropy and the max-entropy as special cases, thus encompassing most quantum entropies in use today. We show several natural properties for this definition, including data-processing inequalities, a duality relation, and an entropic uncertainty relation.   
\end{abstract}

\maketitle

\section{Introduction}

The Shannon entropy~\cite{shannon48} and related measures, like mutual information and relative entropy (also known as Kullback-Leibler divergence), capture many operational quantities in information and communication theory. However, in non-asymptotic or non-ergodic settings, where the law of large numbers does not readily apply, 
other entropy measures typically take over, for example the min-, the max-, or the collision entropy. 
The R\'enyi entropies~\cite{renyi61} nicely unify these different and isolated measures:  there is one (parameterized) entropy measure, the R{\'e}nyi divergence, from which the other measures can be naturally derived. 
Not only is this appealing from a theoretical perspective, but the R\'enyi entropies also have found various applications and are widely used as a technical tool in information theory.

Most of the above mentioned information measures have been generalized to the quantum setting. Most notably, the von Neumann entropy, Renner's (conditional) min- and max-entropies~\cite{renner05}, and a family of R\'enyi relative entropies derived from Petz's quasi-entropy~\cite{petz86} are well-studied and have found various applications. 
Nevertheless, the situation in the quantum setting is much less satisfactory in that these generalizations are (partly) incompatible with each other. For instance, whereas the classical conditional min-entropy can be naturally derived from the R{\'e}nyi divergence, this does not hold for their quantum counterparts. 

To this end, we propose a new quantum generalization of the family of R\'enyi entropies. (To the best of our knowledge, this generalization was first discussed in~\cite{mytutorial12} by one of the present authors.) Specifically, we propose a new definition for the quantum R{\'e}nyi divergence $\Dnew_{\alpha}(\rho\|\sigma)$ for $\alpha \in [\demi,1) \cup (1,\infty)$. From our new definition, we can naturally derive a new notion of conditional R\'enyi entropy $\Hnew_{\alpha}(A|B)_{\rho}$. This new notion contains the various isolated (conditional) entropy measures as special cases. 
Thus, as in the classical case, we have one entropy measure from which most entropies in use today can be naturally derived as special or limiting cases.

We believe that our quantum R\'enyi entropies constitutes a powerful generalization of the classical R\'enyi entropies that will find significant application in quantum information theory. This is supported by the fact that these entropies have several natural properties, which we briefly discuss here. 

\begin{itemize}
\item {\em Data processing}: The divergence $\Dnew_\alpha(\rho\|\sigma)$ can only decrease when acting on $\rho$ and $\sigma$ (by a completely positive trace preserving map). Similarly, the conditional entropy $\Hnew_{\alpha}(A|B)_{\rho}$ can only increase when acting on $B$. 

\item {\em Monotonicity in $\alpha$}: The divergence $\Dnew_{\alpha}(\rho\|\sigma)$ is monotonically increasing, and, correspondingly, $\Hnew_{\alpha}(A|B)_{\rho}$ is monotonically decreasing in $\alpha$. 

\item {\em Duality}: For any pure state $\rho_{ABC}$, and for $\alpha$ and $\beta$ with $\frac{1}{\alpha} + \frac{1}{\beta} = 2$, the conditional entropy satisfies $\Hnew_{\alpha}(A|B)_{\rho} = -\Hnew_{\alpha}(A|C)_{\rho}$. 
\end{itemize}

We give proofs that our quantum R\'enyi entropies satisfy these properties; for the data processing property, our proof only applies to $\alpha \in (1,2]$, but it does hold for arbitrary $\alpha \geq \frac12$ (see below).

\paragraph*{Related work:} This paper is an update to~\cite{lennert13a}, which introduces these new entropies, but which has several of the properties of the new entropies stated as conjectures.
The main contribution of this update is that we have resolved the monotonicity in $\alpha$ and the duality for the conditional entropy. Inspired by~\cite{lennert13a}, and concurrent to our work our conjectures were approached in two further independent contributions: Frank and Lieb~\cite{frank13} prove data-processing for arbitrary $\alpha\geq\frac 12$, and Beigi~\cite{beigi13new} proves monotonicity in $\alpha$, duality for the conditional entropy, and data-processing for $\alpha > 1$. As such, all conjectures of~\cite{lennert13a} have now been resolved.

We also point out that after the completion of~\cite{martinthesis} and concurrently with~\cite{lennert13a}, Wilde, Winter and Yang~\cite{wilde13} employed the same notion of quantum R\'enyi divergence under the name ``sandwiched quantum R\'enyi relative entropy'', and they independently achieved some of the results presented in~\cite{lennert13a}. 
In addition, their work provides a first application of the R\'enyi divergence to solve an important open problem in quantum information theory.
Most recently, Mosonyi and Ogawa~\cite{MO13} have found an operational interpretation of the quantum R\'enyi divergence for $\alpha > 1$ as a generalized cutoff rate in the strong converse problem of hypothesis testing.

\paragraph*{Outline:} We first motivate our definition of quantum R{\'e}nyi divergence in Section~\ref{sec:axiomatic} and then discuss its properties in Section~\ref{sec:overview-prop}. In Section~\ref{sec:overview-cond} we consider the resulting notion of conditional R\'enyi entropies and examine their properties. All proofs are deferred to Section~\ref{sec:proofs}.

\section{An Axiomatic Approach to Quantum R\'enyi Entropies}
\label{sec:axiomatic}

\subsubsection{Quantum R\'enyi Entropies}

Alfr\'ed R\'enyi, in his seminal 1961 paper~\cite{renyi61}, based on previous work by Feinstein and Fadeev, 
investigated an axiomatic approach to derive the Shannon entropy~\cite{shannon48}. He found that five natural requirements for functionals on a probability space single out the Shannon entropy, and by relaxing one of these requirements, he found a family of entropies now named after him. The requirements can be readily generalized to the quantum setting.
For this purpose, let us denote by $\cS$ the set of sub-normalized quantum states, i.e.\ $\rho \in \cS$ is positive semi-definite (denoted $\rho \geq 0$) and has trace $\tr[\rho] \in (0, 1]$.  For our definitions, we follow the convention that for any function $f$ diverging at $0$ we set $f(0) = 0$. In particular, for $\rho \geq 0$, $\log \rho$ and $\rho^{-1}$ are only evaluated on their support.

We are interested in a functional $H(\cdot): \cS \to \mathbb{R}$ satisfying the following properties:
\begin{enumerate}
\item \textbf{Continuity:} $H(\rho)$ is continuous in $\rho \in \cS$.
\item \textbf{Unitary invariance:} $H(\rho) = H(U\rho U^{\dagger})$ for any unitary $U$.
\item \textbf{Normalization:} $H(\frac12) = \log 2$.
\item \textbf{Additivity:} $H(\rho \otimes \tau) = H(\rho) + H(\tau)$ for all $\rho,\tau \in \cS$.
\item \textbf{Arithmetic Mean:} $H( \rho \oplus \tau ) = \frac{\tr[\rho]}{\tr[\rho+\tau]} \cdot H( \rho) + \frac{\tr[\tau]}{\tr[\rho+\tau]} \cdot H( \tau)$ for $\rho,\tau \geq 0$ with $\tr[\rho+\tau] \leq 1$.
\end{enumerate}

Indeed, the von Neumann entropy $H(\rho) := -\tr\big[\rho \log\rho\big]/\tr[\rho]$ satisfies (i)-(v). On the other hand, following R\'enyi's argument~\cite[Thm.\ 1]{renyi61}, we find that (i)-(iv) enforce $H(\lambda) = \log \frac{1}{\lambda}$ for any $\lambda \in (0,1]$, the function thus evaluates what Shannon called the \emph{surprisal} of an event occurring with probability $\lambda$. Particularly, the normalization~(iii) enforces that the logarithm is taken with regards to a particular basis that we leave unspecified here. (Moreover, throughout this paper, $\exp$ is the inverse of $\log$.) Property~(v) then ensures that the arithmetic mean of the surprisal is considered.
We thus find that the Shannon entropy is the unique functional satisfying the classical specializations of (i)-(v) and its unique quantum generalization with unitary invariance~(ii) is the von Neumann entropy.

However, there is no a priori reason why one should only consider the arithmetic mean of the surprisal. R\'enyi thus replaced (v) with a different requirement, namely
\begin{enumerate}
\item[(v')] \textbf{General Mean:} There exists a continuous and strictly monotonic function $g$ such that, for $\rho,\tau \geq 0$ with $\tr[\rho+\tau] \leq 1$,
\begin{align*}
  H(\rho \oplus \tau) = g^{-1} \bigg( \frac{\tr[\rho]}{\tr[\rho+\tau]} \cdot g\big( H( \rho) \big) + \frac{\tr[\tau]}{\tr[\rho+\tau]} \cdot g \big( H( \tau) \big)\bigg) .
\end{align*}
\end{enumerate}
and shows~\cite[Thm.\ 2]{renyi61} that the \emph{R\'enyi entropy} of order $\alpha$ satisfies (i)-(iv) and (v') with $g_{\alpha}(x) = \exp\bigl( (1-\alpha) x\bigr)$. In the quantum setting, the \emph{R\'enyi entropy} of order $\alpha \in (0, 1) \cup (1, \infty)$ is 
given  as  
\begin{align}
  \Hold_{\alpha}(\rho) := \frac{1}{1-\alpha} \log \frac{ \tr[\rho^{\alpha}]}{\tr[\rho]} \,. \label{eq:renyi}
\end{align}

The following observations are obvious from the form of the function $g_{\alpha}(x) = \exp\big((1-\alpha)x\big)$ used to evaluate the average. The larger $\alpha$ is the more weight will be put on contributions with small surprisal. For $\alpha > 1$ contributions with less surprisal are preferred and for $\alpha < 1$ the opposite is true. From this follows that the entropies are monotonically decreasing for increasing~$\alpha$. We have $H_1(\rho) := \lim_{\alpha\nearrow 1} \Hold_{\alpha}(\rho) = \lim_{\alpha\searrow 1} \Hold_{\alpha}(\rho) = H(\rho)$ as a continuous extension. We can also extend the definition to the limit $\alpha\to\infty$, where we obtain the
 \emph{min-entropy}
\begin{align}
  H_{\min}(\rho) := \lim_{\alpha\to \infty} \Hold_{\alpha}(\rho) = - \log \| \rho \|  \label{eq:min},
\end{align}
where $\|\cdot\|$ denotes the operator norm. It is easy to verify that the min-entropy indeed satisfies (i)--(iv); however, the mean
property (v') must be generalized to allow for the relation $H_{\min}(\rho \oplus \sigma) = \min \{ H_{\min}(\rho),\, H_{\min}(\sigma) \}$~[45]. The Hartley entropy $H_0(\rho) := \lim_{\alpha \to 0} \Hold_{\alpha}(\rho)$ is sometimes defined but does not satisfy our stringent continuity condition~(i) as it jumps when the rank of $\rho$ changes. (More precisely, we expect $\lim_{\epsilon \searrow 0} H(\rho \oplus \epsilon) = H(\rho)$ for $\rho > 0$.)

\subsubsection{Quantum Divergences}

R\'enyi then applies this axiomatic approach to divergences or relative entropies, i.e.\ functionals $D(\cdot\|\cdot)$ that map a pair of operators $\rho,\sigma \geq 0$ with $\rho \neq 0$, $\sigma \gg \rho$ onto the real line. Here, $\sigma \gg \rho$ denotes the fact that $\sigma$ dominates $\rho$, i.e.\ that the kernel of $\sigma$ is contained in the kernel of $\rho$.
Again, the six axioms naturally translate to the quantum setting as follows:
\begin{enumerate}
\item[(I)] \textbf{Continuity:} $D(\rho\|\sigma)$ is continuous in $\rho,\sigma \geq 0$, wherever $\rho \neq 0$ and $\sigma \gg \rho$.
\item[(II)] \textbf{Unitary invariance:} $D(\rho\|\sigma) = D(U\rho U^{\dagger}\|U \sigma U^{\dagger})$ for any unitary $U$.
\item[(III)] \textbf{Normalization:} $D(1\|\frac12) = \log 2$.
\item[(IV)] \textbf{Order:} If $\rho \geq \sigma$ (i.e.\ if $\rho - \sigma \geq 0$ is positive semi-definite), then $D(\rho\|\sigma) \geq 0$. And, if $\rho \leq \sigma$, then $D(\rho\|\sigma) \leq 0$.
\item[(V)] \textbf{Additivity:} $D(\rho \otimes \tau\|\sigma \otimes \omega) = D(\rho\|\sigma) + D(\tau\|\omega)$ for all $\rho,\sigma,\tau,\omega \geq 0$ such that $\sigma \gg \rho, \omega \gg \tau$.
\item[(VI)] \textbf{General Mean:} There exists a continuous and strictly monotonic function $g$ such that, for all $\rho,\sigma,\tau,\omega \geq 0$ with $\rho \neq 0$, $\tau \neq 0$, $\sigma \gg \rho$, $\omega \gg \tau$, $\tr[\rho+\tau] \leq 1$ and $\tr[\sigma+\omega] \leq 1$,
\begin{align*}
  D(\rho \oplus \tau \| \sigma \oplus \omega) = g^{-1} \bigg( \frac{\tr[\rho]}{\tr[\rho+\tau]} \cdot g\big( D( \rho \|\sigma) \big) + \frac{\tr[\tau]}{\tr[\rho+\tau]} \cdot g \big( D( \tau \| \omega) \big)\bigg) .
\end{align*}
\end{enumerate}

Again, R\'enyi~\cite[Thm.~3]{renyi61} first shows that (I)--(V) imply $D(\lambda\|\mu) = \log \frac{\lambda}{\mu}$ for two scalars $\lambda, \mu > 0$, a quantity that is often referred to as the log-likelihood ratio. He then considers general continuous and strictly monotonic functions to define a mean in~(VI), as long as they are compatible with (I)--(V). 
Assuming for the moment that $\rho$ and $\sigma$ commute, R\'enyi shows that Properties~(I)--(VI) are satisfied only if $g$ is either linear or exponential \cite[Thm.~3]{renyi61}. The former leads to the so called \emph{Kullback-Leibler divergence} while the latter yields the \emph{R\'enyi divergence for $\alpha \in (0, 1) \cup (1, \infty)$}, which are respectively given as
\begin{align}
  D_1(\rho\|\sigma) =  \frac{\tr\big[\rho (\log \rho - \log \sigma) \big]}{\tr[\rho]} \qquad \textrm{and} \qquad 
  \Dold_{\alpha}(\rho\|\sigma) = \frac{1}{\alpha-1} \log \frac{ \tr[ \rho^{\alpha} \sigma^{1-\alpha} ] }{\tr [ \rho ]}   \label{eq:relative} \,. 
\end{align}

Note that values of $\alpha \leq 0$ are excluded only due to the continuity requirement, (I).
In the following, we are concerned with generalizing these divergences to non-commuting operators. We emphasize at this point that the axiomatic approach presented above does not uniquely determine a single quantum generalization of the R\'enyi divergence of order $\alpha$. In particular, the ordering of operators $\rho$ and $\sigma$ in $\Dold_{\alpha}(\rho\|\sigma)$ is not uniquely determined by the axioms, although Property (IV) enforces some constraints. Furthermore, the following observation is crucial. For commuting operators, R\'enyi's axioms imply, via the explicit expressions, that the divergences satisfy a data-processing inequality, i.e.\ the functional $D(\cdot\|\cdot)$ is contractive under application of trace-preserving completely positive maps to both arguments~(see, e.g.,~\cite{csiszar95}). It remains open whether this implication also holds for the non-commutative case. Instead, we find the following observation useful, which relates joint convexity~[46] resp.\ concavity and the data-processing inequality. It establishes that for functionals satisfying our axioms these properties are equivalent.

\begin{prop}
\label{pr:dp-jc}
  Let $D(\cdot\|\cdot)$ be a functional satisfying (I)--(VI) and let  $g$ be as in (VI). Then, the following two statements are equivalent.
  \begin{enumerate}
    \item[(1)] The functional $g(D(\cdot\|\cdot))$ is jointly convex on normalized states if $g$ is monotonically increasing or jointly concave on normalized states if $g$ is monotonically decreasing.
    \item[(2)] $D(\cdot\|\cdot)$ satisfies the data-processing inequality.
  \end{enumerate}
\end{prop}

We point out that the axioms enforce that $g(D(\cdot\|\cdot))$ is unitarily invariant and that $g(D(\rho \otimes \tau \| \sigma \otimes \tau)) = g(D(\rho\|\tau))$. The implication $(1) \!\implies\! (2)$ then follows by a now standard argument originating from the study of the relative entropy: It was shown by Uhlmann~\cite{UhlEndl,UhlRel} that convexity implies monotonicity under the partial trace. Lindblad discusses the concept of relative entropy in~\cite{LinCom2} and shows that it is monotone under the partial trace. In~\cite{LinCom} he
uses the Stinespring representation theorem~\cite{Stine} to show that this implies the data-processing inequality for quantum channels. Similarly, joint concavity implies contractivity of $-g(D(\cdot\|\cdot))$  and, thus, $D(\cdot\|\cdot)$. See also~\cite{MBRRev} for a review of the topic and \cite[Thm. 5.16]{wolf-ln} for a proof of the implication $(1) \!\implies\! (2)$.

We provide a proof that the converse is also true, namely that $(2) \!\implies\! (1)$ for functionals satisfying the above axioms.

\begin{proof}[Proof of $(2) \!\implies\! (1)$]
   
      Consider normalized operators $\rho, \sigma, \tau, \omega \geq 0$ and $\lambda \in [0, 1]$. Then, due to data-processing
   \begin{align*}
   D\big( \lambda \rho + (1-\lambda) \tau \big\| \lambda \sigma + (1-\lambda)\omega \big)
   \leq D\big( \lambda \rho \oplus (1-\lambda) \tau \big\| \lambda \sigma \oplus (1-\lambda)\omega \big)
   \end{align*}
   If $g$ is increasing, we find that
  \begin{align*}
     &g\big(D\big( \lambda \rho + (1-\lambda) \tau \big\| \lambda \sigma + (1-\lambda)\omega \big)\big)\\
   &\qquad \leq g\big(D\big( \lambda \rho \oplus (1-\lambda) \tau \big\| \lambda \sigma \oplus (1-\lambda)\omega\big) \big)  \\
   &\qquad = \lambda {g}\big(D(\lambda \rho\|\lambda \sigma)\big) + (1-\lambda) {g}\big(D( (1-\lambda) \tau \| (1-\lambda)\omega )\big) \\
   &\qquad = \lambda {g}\big(D(\rho\|\sigma)\big) + (1-\lambda) {g}\big(D(\tau \| \omega)\big) \,,
  \end{align*}   
  where we used property~(VI) for the first equality and (V) and (IV) for the last. It follows that $g(D(\cdot\|\cdot)$ is jointly convex. An analogous argument yields joint concavity if $g$ is decreasing.
\end{proof}

The Kullback-Leibler divergence can readily be extended to the non-commuting quantum setting where it is usually called \emph{relative entropy}.

\begin{defi}[Quantum Relative Entropy]
  Let $\rho, \sigma \geq 0$ with $\rho \neq 0$. The \emph{quantum relative entropy} of $\rho$ and $\sigma$ is defined as
  \begin{align*}
    \Dold_1(\rho\|\sigma) := \begin{cases}  \frac{1}{\tr[\rho]} \tr\big[\rho (\log \rho - \log \sigma) \big] & \textrm{if}\ \sigma \gg \rho \\ \infty & \textrm{if}\ \sigma \not\gg \rho \end{cases} \quad.
  \end{align*}
\end{defi}
We refer to \cite{MBRRev} for a review of properties of this quantity. In particular, the definition in~\eqref{eq:relative} satisfies properties (I)--(VI) for non-commuting $\rho$ and $\sigma$. 

For the R\'enyi relative entropy~\eqref{eq:relative}, we note again that the ordering of the operators $\rho$ and $\sigma$ is relevant and \emph{not unique.} This has been noted, for example, by Ogawa and Hayashi~\cite{ogawa04}.
The expression $\Dold_{\alpha}(\cdot\|\cdot)$ with the ordering as in~\eqref{eq:relative}, for general non-commuting $\rho$ and $\sigma$, has been considered as a quantum generalization of the relative R\'enyi entropy. It was studied for example in \cite{MD09,HMO08,koenig09b,tomamichel08,ogawa00} and the term ``generalized R\'{e}nyi relative entropy'' occurs for example in \cite{MD09}. $\Dold_{\alpha}(\cdot\|\cdot)$ is related to Petz's quasi-entropies~\cite{petz86}, i.e.\ $g(\Dold_{\alpha}(\cdot\|\cdot))$ is the quasi-entropy corresponding to the function $t\mapsto t^{\alpha}$, and it satisfies our axioms as well as data-processing in the range $\alpha \in (0, 1) \cup (1, 2)$. The latter follows from the operator concavity resp.\ convexity of $t \mapsto t^{\alpha}$ in this range.

$\Dold_{\alpha}(\cdot\|\cdot)$ has proven to be a useful tool in some derivations (see, e.g.,~\cite{koenig09b,tomamichel08,ogawa00}) and it also plays a prominent role in quantum hypothesis testing (see, e.g.,~\cite{audenaert12} and references therein). In particular, the Chernoff and Hoeffding distances (cf.~\cite[Thm.~1.1 and Thm.~1.3]{audenaert12}) can be seen as optimizations of a function $\alpha \mapsto f(\alpha, \Dold_{\alpha})$ over a range of $\alpha$.
In~\cite{MH11}, Mosonyi and Hiai obtain an operational interpretation for $\Dold_{\alpha}(\cdot\|\cdot)$ as a generalized cutoff rate for quantum hypothesis testing.
However, for $\alpha > 1$, it has found little operational significance in quantum information theory so far. The min-, max- and collision entropies occurring in various operational scenarios in quantum information theory (see~\cite{mythesis,renner05} and references therein for an overview) are not specializations of $\Dold_{\alpha}(\cdot\|\cdot)$ for  any value of $\alpha$~\cite[Sec.~2]{tomamichel08} and~\cite[App.~B.2]{mythesis}.
Moreover, for $\alpha > 2$, it is easy to verify that the quantity $g(\Dold_{\alpha}(\cdot\|\cdot))$ is neither convex nor concave in the second argument and thus does not satisfy the data-processing inequality according to Proposition~\ref{pr:dp-jc}.

In the following, we consider an alternative generalization of the R\'enyi divergence.

\section{Overview of Results}
\label{sec:overview}

\subsection{A New Quantum R\'enyi Divergence}
\label{sec:overview-prop}

In this work, we propose a different non-commutative generalization of the R\'enyi divergence. For better readability, all proofs are deferred to Section~\ref{sec:proofs-axioms}.

\begin{defi}[Quantum R\'enyi Divergence]\label{df:renyi}
  Let $\rho,\sigma \geq 0$ with $\rho \neq 0$. Then, for any $\alpha \in (0, 1) \cup (1, \infty)$, the \emph{order-$\alpha$ R\'enyi divergence} of $\rho$ and $\sigma$ is defined as
  \begin{align*}
    \Dnew_{\alpha}(\rho\|\sigma) := \begin{cases} \frac{1}{\alpha-1} \log \bigg( \frac{1}{\tr[\rho]} \tr\Big[ \big( \sigma^{\frac{1-\alpha}{2\alpha}} \rho \sigma^{\frac{1-\alpha}{2\alpha}}\big)^{\alpha} \Big] \bigg)& \textrm{if}\ \rho \not\perp \sigma \land ( \sigma \gg \rho \lor \alpha < 1 ) \\ \infty & \textrm{else} \end{cases}\quad .
  \end{align*}
\end{defi}
  Here, $\rho \perp \sigma$ denotes the condition that $\rho$ and $\sigma$ are orthogonal (in particular $0 \perp \rho$ for all $\rho \geq 0$). Note that $\sigma \gg \rho$ implies $\rho \not\perp \sigma$, but non-orthogonality is in fact sufficient to make sure the quantity is finite when $\alpha < 1$.

Let us first verify that this definition indeed satisfies~(I)--(VI). 
\begin{thm} \label{thm:axioms}
  Definition~\ref{df:renyi} satisfies Properties~(I)--(VI) for $\alpha \in [\frac12, 1) \cup (1, \infty)$.
\end{thm}

In the following subsections we discuss several properties of $\Dnew_{\alpha}(\cdot\|\cdot)$.

\subsubsection{First properties}

The notion of \emph{divergence} requires that the quantity be positive definite, which \emph{for commuting $\rho$ and $\sigma$} is well-known~\cite{csiszar95}. We can show the following statement.

\begin{thm}
  \label{thm:positive}
  Let $\rho, \sigma \geq 0$, $\rho \neq 0$ and $\tr[\rho] \geq \tr[\sigma]$. Then,
  we have $\Dnew_{\alpha}(\rho\|\sigma) \geq 0$. Furthermore if $\rho = \sigma$ then $\Dnew_{\alpha}(\rho\|\sigma) = 0$.
\end{thm}

In fact, the latter statement can easily be extended to ``$\Dnew_{\alpha}(\rho\|\sigma) = 0$ 
  if and only if $\rho = \sigma$ for any $\rho,\sigma \geq 0$ with $\tr[\rho] \geq \tr[\sigma]$'' by using the data-processing inequality by Frank and Lieb~\cite{frank13}. If $\rho \neq \sigma$, there exists a measurement $\mathcal{M}$ that allows to distinguish between the two states and thus, for normalized $\rho$ and $\sigma$, we have
  \begin{align*}
    \Dnew_{\alpha}(\rho\|\sigma) \geq \Dnew_{\alpha}\big(\mathcal{M}(\rho)\|\mathcal{M}(\sigma)\big) > 0 ,
  \end{align*}
  where the latter inequality follows from the positive definiteness of the classical R\'enyi divergence~\cite{csiszar95}. This extends to $\rho$, $\sigma$ with general trace, as long as the condition $\tr[\rho] \geq \tr[\sigma]$ is satisfied.
  A similar extension was also shown by Beigi~\cite{beigi13new} and previously by Wilde~\emph{et al.}~\cite{wilde13} for $\alpha \in (1, 2]$.

It is also interesting to compare $\Dold_{\alpha}(\rho\|\sigma)$ to $\Dnew_{\alpha}(\rho\|\sigma)$ for fixed $\rho$ and $\sigma$. Wilde \emph{et al.}~\cite{wilde13} and Datta and Leditzky~\cite{datta13} observed that, as an immediate consequence of the Araki-Lieb-Thirring trace inequality~\cite{ariki90,liebthirring}, we have $\Dnew_{\alpha}(\rho \| \sigma) \leq \Dold_{\alpha}(\rho \| \sigma)$.

We also show the following natural property:
\begin{prop}\label{thm:bigger-sigma}
    Let $\rho$ with $\rho \neq 0$ and let $\sigma' \geq \sigma \geq 0$. Then, for $\alpha \in [\frac12, 1) \cup (1, \infty)$, we have $\Dnew_{\alpha}(\rho \| \sigma') \leq \Dnew_{\alpha}(\rho \| \sigma)$.
\end{prop}
Note that this property is well known for the relative entropy where it follows directly from the operator monotonicity of the logarithm.

\subsubsection{Limits and Important Special Cases}

An analogue of the min-entropy can be defined as a divergence~\cite[Sec.~III]{datta08} and it can be conveniently expressed as a semi-definite optimization problem.
\begin{defi}[Quantum Relative Max-Entropy]\label{df:rel-max}
  Let $\rho, \sigma \geq 0$. The \emph{max relative entropy} of $\rho$ and $\sigma$ is defined as
\begin{align*}
  D_{\max}(\rho\|\sigma) := \inf \{ \lambda \in \mathbb{R} \,|\, \rho \leq \exp(\lambda) \sigma \} \,.
\end{align*}
\end{defi}
Note that this definition in particular implies that $D_{\max}(\rho\|\sigma) = \infty$ if $\rho \neq 0$ and $\sigma \not\gg\rho$.
Our first result shows that the relative max-entropy can be seen as the limit of the $\alpha$-order R\'enyi divergence when $\alpha \to \infty$. 
Also, the quantum relative entropy is the limit of the $\alpha$-order R\'enyi divergence when $\alpha \to 1$. 
\begin{thm} \label{thm:limits}
  Let $\rho, \sigma \geq 0$ with $\rho \neq 0$ and $\sigma \gg \rho$. Then,
  \begin{align*}
    D_{\max}(\rho\|\sigma) &= \lim_{\alpha\to\infty} \Dnew_{\alpha}(\rho\|\sigma) \,,\\
    D(\rho\|\sigma) &= \lim_{\alpha\nearrow 1}  \Dnew_{\alpha}(\rho\|\sigma)  = \lim_{\alpha \searrow 1} \Dnew_{\alpha}(\rho\|\sigma) \,.
  \end{align*}
\end{thm}
The proof is presented in Section~\ref{sec:proof-limits}.

Two other special cases have been noted in the literature. For $\alpha = 2$, we recover the \emph{collision relative entropy} (see~\cite[Def.~5.3.1]{renner05}, where this quantity is defined as a conditional entropy), which has, for example, found applications in randomness extraction and min-entropy sampling~\cite{dupuis13}. It is given as
\begin{align*}
  \widetilde{D}_2(\rho\|\sigma) = \log \frac{1}{\tr[\rho]} \tr\Big[ \big( \sigma^{-\frac{1}{4}} \rho  \sigma^{-\frac{1}{4}} \big)^2 \Big] .
\end{align*}
Moreover, the specialization $\alpha = \frac12$ is related to the fidelity,
\begin{align*}
  \widetilde{D}_\frac{1}{2}(\rho\|\sigma) = -2 \log \frac{1}{\tr[\rho]} \tr \Big[ \Big( \sqrt{\sigma}\, \rho  \sqrt{\sigma}\Big)^{\frac12} \Big] = -2 \log \frac{F(\rho,\sigma)}{\tr[\rho]}\, ,
\end{align*}
where $F(\rho,\sigma) := \| \sqrt{\rho} \sqrt{\sigma} \|_1 = \tr | \sqrt{\rho}\sqrt{\sigma} |$.

The expression $\lim_{\alpha\to0} \Dold_{\alpha}(\rho\|\sigma) = - \log \Tr( \Pi_\rho \sigma )$, where $\Pi_\rho$ denotes the projector to the support of $\rho$ is introduced by Datta~\cite[Def.~2]{datta08}, where it is called the relative min-entropy. This limit is not reproduced by our definition of $\Dnew_{\alpha}(\rho || \sigma)$ in general~\cite{datta13}.

For an overview of the definitions of the entropies used in this article refer to Table~\ref{tb:overview}.

\subsubsection{Joint Convexity/Concavity and Data-Processing}
\label{sec:dp}

Consider a completely positive trace-preserving map (CPTPM) $\cE$. For such maps 
the implication $\rho \geq \sigma \implies \cE(\rho) \geq \cE(\sigma)$ holds. Thus, from the definition of the max relative entropy, we immediately find that the data-processing inequality $D_{\max}(\rho\|\sigma) \geq D_{\max}(\cE(\rho) \| \cE(\sigma) )$ holds. As mentioned before this property also holds for the quantum relative entropy and is closely related to strong sub-additivity~\cite{lieb73}. For $D_{\frac 12}(\cdot\|\cdot)$ data-processing follows directly from the contractivity of the fidelity under CPTPMs.

\begin{thm}[Data-Processing]
  \label{thm:data-proc}
  Let $\rho,\sigma \geq 0$, $\rho \neq 0$ and $\alpha \in (1,2]$. Then, for any CPTPM $\cE$, we have
  \begin{align}
     \Dnew_{\alpha}(\rho\|\sigma) \geq \Dnew_{\alpha}( \cE(\rho) \| \cE(\sigma) ) \, . \label{eq:data-proc}
  \end{align}
  Moreover, $\exp\bigl((\alpha-1)\Dnew_{\alpha}(\cdot \|\cdot)\bigr)$ is jointly convex.%
\end{thm}

For the proof we only need to establish that $\exp\bigl((\alpha-1)\Dnew_{\alpha}(\cdot \|\cdot)\bigr)$ is jointly convex (see Section~\ref{sec:proof-dp}) and then employ Proposition~\ref{pr:dp-jc}.
The proof, which uses a strategy proposed in~\cite[Thm.~5.16]{wolf-ln}, is deferred to Section~\ref{sec:proof-dp}. 
Whereas our proof only works for $\alpha \in (1,2]$, the data processing inequality actually holds for all $\alpha \in [\frac12, 1) \cup (1, \infty)$, and  $\exp\big((\alpha-1)\Dnew_{\alpha}(\cdot\|\cdot)\big)$ is jointly convex for all $\alpha \in (1,\infty)$ and jointly concave for $\alpha \in [\demi, 1)$. 
This was recently shown by Frank and Lieb~\cite{frank13}, and independently by Beigi~\cite{beigi13new} (for $\alpha > 1$), resolving our earlier conjecture in~\cite{lennert13a}. 
Finally, note that we found numerical counter-examples for data-processing when $\alpha < \demi$.

\subsubsection{Monotonicity in $\alpha$}

The classical R\'enyi divergences are monotonically increasing in $\alpha$~\cite{csiszar95}. This is evident from the mean property~(VI) which ensures that the larger $\alpha$ the more preference is given to contributions with high log-likelihood ratio. Thus, for commuting $\rho, \sigma \geq 0$ and $\alpha,\beta \in (0,1) \cup (1, \infty)$ such that $\alpha \leq \beta$ we have $\Dnew_{\alpha}(\rho\|\sigma) \leq \Dnew_{\beta}(\rho\|\sigma)$. This property extends to the non-commutative setting.

\begin{thm}[Monotonicity]
 \label{thm:mono}
  Let $\rho,\sigma \geq 0$ and $\rho \neq 0$. Then, $\alpha \mapsto \Dnew_{\alpha}(\rho\|\sigma)$ is monotonically increasing.
\end{thm}

See Section~\ref{sec:proof-mono} for a proof. This result has been derived independently by Beigi~\cite{beigi13new}, resolving our earlier conjecture in~\cite{lennert13a}.

\subsection{From Divergence to Conditional Entropy}
\label{sec:overview-cond}

\begin{table}
\begin{center}
\begin{footnotesize}
\begin{tabular}{| c | l | l |}
\hline
$\alpha$-Range & Divergence & Conditional entropy \\ \hline
$[\frac{1}{2}, 1) \cap (1, \infty)$ & $\Dnew_{\alpha}(\rho || \sigma) = \frac{1}{\alpha-1} \log \Big( \frac{1}{\tr[\rho]} \tr\big[ \big( \sigma^{\frac{1-\alpha}{2\alpha}} \rho \sigma^{\frac{1-\alpha}{2\alpha}}\big)^{\alpha} \big] \Big)$ & $\Hnew_{\alpha}(A|B)_{\rho}$ \\  \hline
$(0, 1) \cap (1, 2]$ & $\Dold_{\alpha}(\rho || \sigma) = \frac{1}{\alpha-1} \log \Big( \frac{1}{\tr[\rho]} \tr \big[ \rho^{\alpha} \sigma^{1-\alpha} \big] \Big)$ & $\Hold_{\alpha}(A|B)_{\rho}$  \\ \hline
$\alpha \to 1$ & $\Dnew_1 \equiv \Dold_1 \equiv D$ & $H_1 \equiv H_1' \equiv H$ \\ \hline
$\alpha \to \infty$ & $\Dnew_{\infty} \equiv D_{\max} \neq D_{\infty}'$ as in~\cite{datta08}. & $\Hnew_{\infty} \equiv H_{\min}$ as in~\cite{renner05}. \\ \hline
$\alpha = \frac{1}{2}$ & $\Dnew_{1/2}(\rho ||\sigma) = - 2 \log \frac{F(\rho,\sigma)}{\Tr \rho} \neq \Dold_{1/2}(\rho\|\sigma)$ & $\Hnew_{1/2} \equiv H_{\max}$ as in~\cite{koenig08}.\\ \hline
$\alpha = 2$ & 
$\Dnew_2(\rho\|\sigma) = \log \big( \frac1{\tr [\rho]} \tr \big[\rho\sigma^{-\frac12}\rho\sigma^{-\frac12}\big] \big) \neq \Dold_2(\rho\|\sigma)$ 
& $\Hnew_2$ is defined in~\cite[Def.~5.3.1]{renner05}. 
 \\ \hline
$\alpha \to 0$ & $\Dold_0 \equiv D_{\min} \neq \Dnew_{0}$ as in~\cite{datta08}, see also~\cite{datta13}. & $\Hold_0$ appears in~\cite{renner05} as max-entropy. \\
\hline
\end{tabular}
\end{footnotesize}
\end{center}

\caption{This table overviews the entropic quantities discussed in this article. Here, $\rho, \sigma \geq 0$ with $\sigma \gg \rho$ and $\rho_{AB} \in \cS_{AB}$ as usual. The conditional entropies are defined as $\Hnew_{\alpha}(A|B)_{\rho} = \sup_{\sigma_B \in\cS_B} - \Dnew_{\alpha}(\rho_{AB} \| \id_A \otimes\, \sigma_B)$ and analogously for $\Hold_{\alpha}$.}
\label{tb:overview}
\end{table}

The divergences can be seen as parent quantities to the ordinary entropies, and for all positive~$\alpha$ and $\rho \in \cS$ it is easy to verify from the above definitions and properties~(V), (IV) and (iii) that
\begin{align}
   \Hold_{\alpha}(\rho) = -\Dold_{\alpha}(\rho\|\id) = \log d - \Dold_{\alpha}(\rho\|\pi) = \Hold_{\alpha}(\pi) - \Dold_{\alpha}(\rho\|\pi), \label{eq:rel}
\end{align}
where $\id$ and $\pi = \id/d$ are respectively the identity and the fully mixed state on the support of $\rho$, and $d$ is the rank of~$\rho$. 
Thus, if we view the R\'enyi divergence as a {\em distance measure} (even though it is not a metric in the mathematical sense), we can understand the R\'enyi entropy $\Hnew_{\alpha}(\rho)$ as the maximal possible entropy (of a state with the same support), which is $\log d$ and attained by the state $\pi$, minus how far away the real state $\rho$ is from $\pi$. 

We now consider bipartite quantum systems and \emph{conditional entropies}.
Let $\rho_{AB} \in \cS_{AB}$ be a bipartite state on $AB$ with $\tr[\rho_{AB}]=1$ and $\rho_B$ its marginal on $B$. The \emph{conditional von Neumann entropy} of $\rho_{AB}$ given $B$ is defined  as $H(A|B)_{\rho} := H(\rho_{AB}) - H(\rho_B)$. 
This can also be written as
\begin{align}
    \nonumber H(A|B)_{\rho} &= H(\rho_{AB}) - H(\rho_B) - \!\inf_{\sigma_B \in \cS_B}\! D_1(\rho_{B}\|\sigma_B)\\
    &= - \!\inf_{\sigma_B \in \cS_B}\! D_1(\rho_{AB}\|\id_A \otimes\, \sigma_B) \label{eq:cond-def} 
\end{align}
due to Klein's inequality~\cite{klein31}, i.e., the positive-definiteness of $D$. 

This approach of defining conditional entropies by optimizing the divergence has proven very fruitful. For example, Renner's conditional min-entropy~\cite[Sec.~3.1.1]{renner05} can be defined via the relation
\begin{align}
  H_{\min}(A|B)_{\rho} := \sup_{\sigma_B \in \cS_B} - D_{\max}(\rho_{AB} \| \id_A \otimes\, \sigma_B) \,. \label{eq:cond-min}
\end{align}
and the \emph{conditional max-entropy}~\cite[Def.~2 and Thm.~3]{koenig08} is given as
\begin{align}
  H_{\max}(A|B)_{\rho} := \sup_{\sigma_B \in \cS_B} - D_{\frac12}(\rho_{AB} \| \id_A \otimes\, \sigma_B) \,.
 \label{eq:cond-max}
\end{align}
It is natural to generalize this definition to conditional R\'enyi entropies.
\begin{defi}[Quantum Conditional R\'enyi Entropy] Let $\rho_{AB} \in \cS_{AB}$ and $\alpha \in (0, 1) \cup (1, \infty)$. The conditional R\'enyi-entropy of order $\alpha$ of $\rho_{AB}$ given $B$ is defined as
\begin{align*}
  \Hnew_{\alpha}(A|B)_{\rho} := \sup_{\sigma_B \in \cS_B} -\Dnew_{\alpha}(\rho_{AB} \| \id_A \otimes\, \sigma_B) .
\end{align*} 
\end{defi}
Note that the $\sigma_B$ we optimize over constitute a compact set and that the function $\sigma_B \mapsto \Dnew_{\alpha}(\rho_{AB}\|\id_A \otimes \sigma_B)$ is continuous except where it diverges to $+\infty$. Thus, the supremum is finite and attained 
for at least one element of the set. Let $\sigma_B^*$ be an element that achieves the supremum. It is easy to verify that $\tr[\sigma_B^*] = 1$. We also have that $\sigma_B^* \gg \rho_B$ if $\alpha > 1$ and $\sigma_B^* \ll \rho_B$ if $\alpha < 1$.

Furthermore, similar to the interpretation of the unconditional R\'enyi entropy by means of~\eqref{eq:rel}, writing $\Hnew_{\alpha}(A|B)_{\rho} = \log d_A - \inf_{\sigma_B} \Dnew_{\alpha}(\rho_{AB} \| \pi_A \otimes\, \sigma_B)$ where $\pi_A$ is the fully mixed state on the support of $\rho_A$ and $d_A$ is its rank, we can understand the conditional R\'enyi entropy as the maximal possible entropy $\log d_A$, minus how far away (in terms of R\'enyi divergence) the real state $\rho_{AB}$ is from a state that has maximal entropy, which is a state of the form $\pi_A \otimes\, \sigma_B$, as can easily be verified.

\subsubsection{Data-Processing and Chain Rule}

We briefly point out some properties of this notion of conditional R\'enyi entropy. 
The data processing inequality for the R\'enyi divergence immediately translates to the data processing inequality for the conditional R\'enyi entropy: for $\alpha \in [\frac12,1) \cup (1,\infty)$, and for any $\rho_{AB} \in \cS_{AB}$ and any CPTPM $\cE_{B \to B'}$ from $B$ to $B'$ it holds that
$$
\Hnew_{\alpha}(A|B)_{\rho} \leq \Hnew_{\alpha}(A|B')_{\tau}
$$
where $\tau_{AB'}$ is obtained by applying $\cE_{B \to B'}$ to (the $B$-part of) $\rho_{AB}$. 
This in particular implies that $\Hnew_{\alpha}(A|B)_{\rho} \geq \Hnew_{\alpha}(A|BC)_{\rho}$ for any tripartite state $\rho_{ABC} \in \cS_{ABC}$, i.e., conditioning on more can only reduce the entropy. On the other hand, the {\em chain rule} below bounds the amount by which the entropy can drop. 

\begin{prop}[Chain Rule]
For $\alpha \in (0, 1) \cup (1, \infty)$, and for any $\rho_{ABC} \in \cS_{ABC}$, it holds that
$$
\Hnew_{\alpha}(A|BC)_{\rho} \geq \Hnew_{\alpha}(AC|B) - \log d_C
$$
where $d_C$ is the rank of $\rho_C$. 
\end{prop}
The proof is identical to the corresponding proof for $H_{\min}$ due to Renner~\cite{renner05}. 
\begin{proof}
Let $\sigma_B$ with $\Tr[\sigma_{B}] = 1$ be such that $\Hnew_{\alpha}(AC|B)_{\rho} = - \Dnew_{\alpha}(\rho_{ABC} \| \id_{AC} \otimes\, \sigma_B)$. Setting $\pi_C = \id_C/d_C$, we immediately obtain that 
$$
\Hnew_{\alpha}(A|BC)_{\rho} \geq - \Dnew_{\alpha}(\rho_{ABC} \| \id_{A} \otimes\, \sigma_B \otimes \pi_C) = \Hnew_{\alpha}(AC|B)_{\rho} - \log d_C \, ,
$$
which proves the claim. 
\end{proof}

\subsubsection{Conditioning on Classical Information}

We now analyze the behavior of $\Dnew_{\alpha}$ and $\Hnew_{\alpha}$ when applied to partly classical states. Formally, consider normalized states of the form $\rho_{AY} = \bigoplus_y p_y \rho_A^y$ and $\sigma_{AY} = \bigoplus_y q_y \sigma_A^y$, where $\{p_y\}$ and $\{q_y\}$ are probability distributions, and $\rho_A^y$ and $\sigma_A^y$ are normalized states in $\cS_{A}$ for all $y$. We say that $\rho_{AY}$ and $\sigma_{AY}$ have {\em classical} register $Y$. 
A straightforward calculation using Property~(VI) shows that for two such states $\rho_{AY}$ and $\sigma_{AY}$
$$
\Dnew_{\alpha}(\rho_{AY}\|\sigma_{AY}) = \frac{1}{\alpha-1} \log\sum_y p_y^\alpha q_y^{1-\alpha} \exp\Bigl((\alpha-1) \Dnew_{\alpha}(\rho_A^y\|\sigma_A^y)\Bigr) \, .
$$
In other words, the divergence $\Dnew_{\alpha}(\rho_{AY}\|\sigma_{AY})$ decomposes into the divergences $\Dnew_{\alpha}(\rho_A^y\|\sigma_A^y)$ of the ``conditional states''. This also holds for the conditional entropy, though the derivation is slightly more involved (see Section~\ref{sec:proof-class}).

\begin{prop}
  \label{pr:class}
  Let $\rho_{ABY} = \bigoplus_y p_y \rho_{AB}^y$ with $\tr[\rho_{ABY}] = \tr[\rho_{AB}^y] = 1$ for all $y$. Then,
  \begin{align*}
  \Hnew_{\alpha}(A|BY)_{\rho} 
= \frac{\alpha}{1-\alpha}\log \sum_y p_y \exp\Bigl(\textstyle\frac{1-\alpha}{\alpha} \Hnew_{\alpha}(A|B)_{\rho^y}\Bigr)  \, . 
  \end{align*}
\end{prop}

In the special case of an ``empty'' $B$, we obtain
$$
\Hnew_{\alpha}(A|Y)_{\rho} = \frac{\alpha}{1-\alpha}\log \sum_y p_y \exp\Bigl(\textstyle\frac{1-\alpha}{\alpha} \Hold_{\alpha}(\rho_{A}^y)\Bigr)  \, , 
$$
and when considering a state $\rho_{XY} = \bigoplus_y p_y \rho_X^y$ where also $X$ is classical, meaning that $\rho_X^y = \bigoplus_x p_{x|y}$ for every $y$, we recover the notion of classical conditional R\'enyi entropy 
$$
\Hnew_{\alpha}(X|Y)_{\rho} = \frac{\alpha}{1-\alpha}\log \sum_y p_y \biggl(\sum_x p_{x|y}^\alpha \biggr)^{1/\alpha} 
$$
originally suggested by Arimoto~\cite{arimoto77}.

\subsubsection{Duality Relation}

Conditional entropies satisfy a surprising duality relation in that, for any pure tripartite state $\rho_{ABC} \in \cS_{ABC}$ with $\tr[\rho_{ABC}] = 1$, we have
\begin{align}
  H(A|B)_{\rho} = -H(A|C)_{\rho} \qquad \textrm{and} \qquad H_{\min}(A|B)_{\rho} = -H_{\max}(A|C)_{\rho} \,.
\end{align}
For the von Neumann entropy this follows from the Schmidt-decomposition of pure states. For the min- and max-entropies, the respective property was shown by K\"onig \emph{et al.}~\cite{koenig08}. We prove that these are just the limiting cases of the following duality relation.
\begin{thm}[Duality]\label{thm:dual}
  Let $\alpha, \beta \in (\frac12,1) \cup (1,\infty)$ such that $\frac{1}{\alpha} + \frac{1}{\beta} = 2$ and let $\rho_{ABC} \in \cS_{ABC}$ be pure with $\tr[\rho_{ABC}] = 1$. Then,
    $\Hnew_{\alpha}(A|B)_{\rho} = - \Hnew_{\beta}(A|C)_{\rho}$.
\end{thm}
The proof is presented in Section~\ref{sec:proof-dual}. This result has been derived independently by Beigi~\cite{beigi13new}, resolving our earlier conjecture in~\cite{lennert13a}.

\subsubsection{Uncertainty Relation}

There is a strong link between the above duality relation and entropic uncertainty relations.
In fact, Maassen and Uffink~\cite[Eq.~(11)-(12)]{maassen88} showed that, for $\frac{1}{\alpha} + \frac{1}{\beta} = 2$, the classical probability distributions, $\rho_X$ and $\rho_Y$, corresponding to two rank-$1$ projective measurements on an arbitrary state satisfy
\begin{align*}
  \Hold_{\alpha}(\rho_X) + \Hold_{\beta}(\rho_Y) \geq \log \frac{1}{c}, \quad \textrm{where} \quad c = \max_{x,y} |\langle e_x | f_y \rangle|^2 \,.
\end{align*}
Here, $\{\ket{e_x}\}_x$ and $\{\ket{f_y}\}_y$ denote the eigenvectors of the two measurements.
This result was then extended to a tripartite setting with quantum side information in~\cite{berta10} for von Neumann entropies and in~\cite{tomamichel11} for min- and max-entropies.
The latter proof was then generalized to arbitrary conditional entropies satisfying a duality relation (and certain other properties) by Coles~\emph{et al.}~\cite{colbeck11}.
Since our generalized R\'enyi entropies satisfy these properties, it thus immediately follows that the above are just special cases of the following general uncertainty relation.
\begin{thm}[Uncertainty Relation for Conditional R\'enyi Entropies]
  Let $\rho_{ABC} \in \cS_{ABC}$ with $\tr[\rho_{ABC}] = 1$ and let $\alpha, \beta \in (\demi,1) \cup (1, \infty)$ such that $\frac{1}{\alpha} + \frac{1}{\beta} = 2$. Then, for any two positive operator-valued measures $\{ M_x \}_x$ and $\{ N_y \}_y$, we have
  \begin{align*}
     \Hnew_{\alpha}(X|B)_{\rho} + \Hnew_{\beta}(Y|C)_{\rho} \geq \log \frac{1}{c}, \quad \textrm{where} \quad c := \max_{x,y} \Big\| \sqrt{M_x} \sqrt{N_y} \Big\| \,,
  \end{align*}
  where the post-measurement states are respectively given by
  \begin{align*}
    \rho_{XB} := \bigoplus_x \tr_{AC} [ M_x \rho_{ABC} ] \quad \textrm{and} \quad \rho_{YC} := \bigoplus_y \tr_{AB} [ N_y \rho_{ABC} ] \,.
  \end{align*}
\end{thm}

\section{Proofs}
\label{sec:proofs}

\subsection{Preliminaries}

We assume finite dimensions in the following.
For any positive semi-definite operator $X$, we use the following generalization of the Schatten $p$-norm. Let $p \in (0, \infty)$, then
\begin{align*}
  \| X \|_p := \big( \tr ( X^p ) \big)^{\frac{1}{p}} .
\end{align*}
Moreover, $\|X\|_\infty$ denotes the operator norm. Note that when $p<1$, this is no longer a norm, but we nonetheless extend the definition to this case as we will find it convenient.

\begin{lem}
  \label{lem:extended-norm-duality}
    Let $p,q \in \mathbb{R}\setminus \{0,1\}$ be such that $\frac{1}{p}+\frac{1}{q} = 1$. Then, for any $X \geq 0$,
    \begin{align*} 
        \| X \|_p &= \sup_{Z \geq 0 \atop \tr[Z] \leq 1} \tr[XZ^{\frac{1}{q}}] \quad \textrm{if } p>1, \qquad \textrm{and} \qquad
        \| X \|_p &= \inf_{Z \geq 0 \atop {Z \gg X \atop  \tr[Z] \leq 1}} \tr[XZ^{\frac{1}{q}}] \quad \textrm{if } p<1\, .
    \end{align*}
\end{lem}
\begin{proof}
    If $p>1$, we can use the duality of $p$-norms (see e.g.~\cite[Ex.~IV.2.12]{bhatia97}). This yields
    \begin{align*}
        \| X \|_{p} &= \sup_{\| Y \|_q \leq 1} \left| \tr[XY] \right|.
    \end{align*}
    First, note since $X \geq 0$, one can always choose $Y$ to be positive semidefinite and diagonal in the same basis as $X$ (see e.g.~\cite[Prob.~III.6.14]{bhatia97}). Then, let $Z := Y^q$, so that $\| Y \|_q = \tr[Y^q]^{1/q} = \tr[Z]^{1/q}$. The first part of the claim then follows.

    When $p<1$, $\| \cdot \|_p$ is not a norm and $q < 0$, so we derive the statement ourselves. We will solve the optimization problem $\inf_{\tr[Z]\leq 1} \tr[XZ^{1/q}]$ using Lagrange multipliers, and show that it is equal to the left-hand side. We can write the Lagrangian as
    \[ \mathcal{L} = \tr[XZ^{\frac{1}{q}}] - \mu(\tr[Z] - 1). \]
    Note that again, we can always choose $Z$ to commute with $X$. Since this expression is convex in (every diagonal element of) $Z$, the optimal $(Z,\mu)$ must satisfy
    \begin{equation} \label{eq:diff-XZq}
    \frac{1}{q}Z^{\frac{1}{q}-1}X - \mu\ident = 0.
\end{equation}
Rearranging, we get that $Z^{1/q}X = q\mu Z$. Thus, the optimal value is given by $q\mu \tr[Z] = q\mu$. We therefore only need to find the optimal $\mu$. For this, we isolate $Z$ in \eqref{eq:diff-XZq} and use the condition that $\tr[Z] = 1$ to get that $\mu = \frac{1}{q} \| X \|_p$. This yields the second part of the claim.
\end{proof}

We also define the following auxiliary quantity:
\begin{defi}\label{def:ueber-rel-entropy}
Let $\rho \geq 0$ with $\tr[\rho] = 1$ and let $\alpha \in (0, 1) \cup (1,\infty)$. Then, for any $\sigma, \tau \geq 0$, let
\begin{align*}
    \Dnew_{\alpha}(\rho\|\sigma; \tau) := \begin{cases} \infty & \textrm{if } \alpha > 1 \land \sigma \not\gg \rho^{\demi} \tau^{\frac{\alpha-1}{\alpha}} \rho^{\demi} \\
        -\infty & \textrm{if } \alpha < 1 \land \tau \not\gg \rho^{\demi} \sigma^{\frac{1-\alpha}{\alpha}} \rho^{\demi} \\
  \frac{\alpha}{\alpha-1} \log \tr\big( \rho^{\demi} \sigma^{\frac{1-\alpha}{\alpha}} \rho^{\demi} \tau^{\frac{\alpha-1}{\alpha}} \big) & \textrm{else}
  \end{cases}\,.
\end{align*}
\end{defi}
Note that, using Lemma~\ref{lem:extended-norm-duality}, the R\'enyi divergence can be recovered as 
\begin{align}
  \Dnew_{\alpha}(\rho\|\sigma) = \frac{\alpha}{\alpha-1} \log \Big\| \rho^{\demi} \sigma^{\frac{1-\alpha}{\alpha}} \rho^{\demi} \Big\|_{\alpha} = \sup_{\tau \geq 0 \atop \tr[\tau] \leq 1} \Dnew_{\alpha}(\rho\|\sigma; \tau). \label{eq:ueber-to-rel}
\end{align}

\subsection{Continuity and Proof of Claims in Section~\ref{sec:overview-prop}}
\label{sec:proofs-axioms}

The following property justifies the use of the generalized inverse and ensures that the quantity is continuous when the rank of $Y$ changes.

\begin{lem}
  \label{prop:dalphalimittoll}
  Let $X, Y \geq 0$ with $X \neq 0$ and $\alpha \in
  (0,1) \cup (1,\infty)$.  We have
  \begin{equation}
    \label{eq:consis}
    \Dnew_{\alpha}(X\|Y) = \lim_{\xi \searrow 0}
    \frac{\alpha}{\alpha - 1} \log \left(\Tr\left[ (
      (Y+\xi)^{\frac{1}{2\alpha} - \frac{1}{2}} X (Y+\xi)^{\frac{1}{2\alpha} -
        \frac{1}{2}} )^\alpha \right]/\Tr X\right)^{\frac{1}{\alpha}},
  \end{equation}
  where $Y + \xi$ is short for $Y + \xi \ident$, and the limit exists in
  the weaker sense in which a real valued sequence which is bounded
  from below and not bounded from above and which does not have an
  accumulation point is considered as being convergent to $+\infty$.
\end{lem}

\begin{proof}
  With respect to the decomposition $\HH = \supp Y \oplus \ker Y$, we
  write $X = \footnotesize\Big(\begin{array}{cc} X_0 & Z \\
    Z^* & X_1\end{array}\Big)$ and $Y + \xi =
  \footnotesize\Big(\begin{array}{cc}
    Y_0 + \xi & 0 \\ 0 & \xi
  \end{array}\Big)
$.  Thus,
  \begin{eqnarray}
    \label{eq:expression}
    &&\frac{1}{\alpha - 1} \log \Tr\left[ (
      (Y+\xi)^{\frac{1}{2\alpha} - \frac{1}{2}} X (Y+\xi)^{\frac{1}{2\alpha} -
        \frac{1}{2}} )^\alpha \right] = \\\nonumber
    &&\frac{1}{\alpha -1} \log  \Tr \left[ 
      \begin{pmatrix}
        (Y_0 + \xi)^{\frac{1-\alpha}{2\alpha}}X_0(Y_0+\xi)^{\frac{1-\alpha}{2\alpha}} &
        \xi^{\frac{1-\alpha}{2\alpha}}(Y_0+\xi)^{\frac{1-\alpha}{2\alpha}}Z\\
        \xi^{\frac{1-\alpha}{2\alpha}}Z^*(Y_0+\xi)^{\frac{1-\alpha}{2\alpha}} & \xi^{\frac{1}{\alpha}-1}X_1
      \end{pmatrix}^\alpha
      \right] .
  \end{eqnarray}
  Consider first the case $\alpha \in (0,1)$. Notice that
  $\xi^{\frac{1-\alpha}{2\alpha}}$ goes to zero as $\xi \searrow 0$. By
  picking a basis in which $Y_0$ is diagonal, we see that $(Y_0 +
  \xi)^{\frac{1-\alpha}{2\alpha}}$ converges to $Y_0^{1/2\alpha -1/2}$ as $\xi
  \searrow 0$. Hence,
  \begin{eqnarray*}
    &&
    \begin{pmatrix}
      (Y_0 + \xi)^{\frac{1-\alpha}{2\alpha}}X_0(Y_0+\xi)^{\frac{1-\alpha}{2\alpha}} &
      \xi^{\frac{1-\alpha}{2\alpha}}(Y_0+\xi)^{\frac{1-\alpha}{2\alpha}}Z\\
      \xi^{\frac{1-\alpha}{2\alpha}}Z^*(Y_0+\xi)^{\frac{1-\alpha}{2\alpha}} &
      \xi^{\frac{1-\alpha}{\alpha}}X_1
    \end{pmatrix} \longrightarrow 
    \begin{pmatrix}
      Y_0^{\frac{1-\alpha}{2\alpha}}X_0 Y_0^{\frac{1-\alpha}{2\alpha}} &
      0\\
      0 & 0
    \end{pmatrix} .
  \end{eqnarray*}
  as $\xi \searrow 0$.
  Moreover, the eigenvalues of a Hermitian operator depend continuously
  on the operator, and hence  $\Tr Z^\alpha = \sum_j \lambda_j(Z)^\alpha$ is a continuous function of $Z$. 
  Thus, if $X_0 = 0$ the expression~\eqref{eq:expression} goes to
  $+\infty$ and if $X_0 \neq 0$ it goes to $\Dnew_{\alpha}(X_0\|Y_0)$.

  Now suppose $\alpha >1$. If $\supp Y \supseteq \supp X$, then, $X_1$ and $Z$ vanish, and hence~\eqref{eq:expression} becomes $\Dnew_{\alpha}(X_0\|Y_0)$ in the
  limit $\xi \searrow 0$. 
  If, however, $\supp Y \nsupseteq \supp X$,
  then $X_1 \neq 0$. We observe that
  \begin{eqnarray*}
    &&\Tr \left[ 
      \begin{pmatrix}
        (Y_0 + \xi)^{\frac{1-\alpha}{2\alpha}}X_0(Y_0+\xi)^{\frac{1-\alpha}{2\alpha}} &
        \xi^{\frac{1-\alpha}{2\alpha}}(Y_0+\xi)^{\frac{1-\alpha}{2\alpha}}Z\\
        \xi^{\frac{1-\alpha}{2\alpha}}Z^*(Y_0+\xi)^{\frac{1-\alpha}{2\alpha}} & \xi^{1/\alpha-1}X_1
      \end{pmatrix}^\alpha
      \right]\\
		&&=
      \xi^{1-\alpha}  \Tr \left[ 
      \begin{pmatrix}
        (\xi(Y_0 + \xi)^{-1})^{-\frac{1}{2\alpha} + \frac{1}{2}}\, X_0\, (\xi(Y_0 +
        \xi)^{-1})^{-\frac{1}{2\alpha} + \frac{1}{2}} &
        (\xi(Y_0 + \xi)^{-1})^{-\frac{1}{2\alpha} + \frac{1}{2}}\,Z\\
        Z^*\, (\xi(Y_0 + \xi)^{-1})^{-\frac{1}{2\alpha} + \frac{1}{2}}& X_1
      \end{pmatrix}^\alpha
      \right]
  \end{eqnarray*}
  diverges to $+\infty$ in the weak sense as $\xi \searrow
  0$. Indeed, since $\xi(Y_0 + \xi)^{-1} \longrightarrow 0$ and
  $\frac{\alpha-1}{2\alpha} \in [0,1/2)$ a similar continuity argument as in
  the case $\alpha \in (0,1)$ implies that
  \begin{eqnarray*}
    \begin{pmatrix}
      (\xi(Y_0 + \xi)^{-1})^{\frac{\alpha-1}{2\alpha}}\, X_0\, (\xi(Y_0 +
      \xi)^{-1})^{\frac{\alpha-1}{2\alpha}} &
      (\xi(Y_0 + \xi)^{-1})^{\frac{\alpha-1}{2\alpha}}\,Z\\
      Z^*\, (\xi(Y_0 + \xi)^{-1})^{\frac{\alpha-1}{2\alpha}}& X_1
    \end{pmatrix} \longrightarrow 
      \begin{pmatrix}
    0&0\\
    0&X_1
  \end{pmatrix}
  \end{eqnarray*}
  as $\xi \searrow 0$. Hence, the term involving the trace converges to $\Tr
  X_1^\alpha$. Since $X_1 \neq 0$, we conclude that this converging term
  is positive and bounded away from zero and infinity for small
  $\xi$. The statement follows since the prefactor $\xi^{1-\alpha}$
  diverges.
\end{proof}

We are now ready to prove Theorem~\ref{thm:axioms}.

\begin{proof}[Proof of Theorem~\ref{thm:axioms}]
  Continuity~(I) in $\rho$ and $\sigma$ is trivial except for the use of the generalized inverse in our definition. However, Lemma~\ref{prop:dalphalimittoll} shows that our definition is just a continuous extension of the definition restricted to $Y > 0$, which is evidently continuous.  Unitary Invariance~(II) follows from definition and (III) is obviously satisfied. 
  
  The Order relation~(IV) is shown as follows. First, note that due to the operator monotonicity of the function $t \mapsto t^{\beta}$ for $\beta \in (0, 1]$~(see Bhatia~\cite[Thm.~V.1.9]{bhatia97}), we have the following: $\rho \geq \sigma$ implies $\rho^{\frac12}\sigma^{\frac{1-\alpha}{\alpha}} \rho^{\frac12} \geq \rho^{\frac{1}{\alpha}}$ if $\alpha > 1$ and $\rho^{\frac12}\sigma^{\frac{1-\alpha}{\alpha}} \rho^{\frac12} \leq \rho^{\frac{1}{\alpha}}$ if $\alpha \in [\frac12,1)$. Thus, employing the Schatten norm of order $\alpha > 1$,
\begin{align}
  \tr\Big[ \big( \sigma^{\frac{1-\alpha}{2\alpha}} \rho \sigma^{\frac{1-\alpha}{2\alpha}}\big)^{\alpha} \Big] = \Big\| \sigma^{\frac{1-\alpha}{2\alpha}} \rho \sigma^{\frac{1-\alpha}{2\alpha}}\Big\|_{\alpha}^{\alpha}
  = \Big\| \rho^{\frac12}\sigma^{\frac{1-\alpha}{\alpha}} \rho^{\frac12} \Big\|_{\alpha}^{\alpha} \geq \big\| \rho^{\frac{1}{\alpha}} \big\|_{\alpha}^{\alpha} = \tr[\rho] . \label{eq:bb1}
\end{align}
Hence, $\Dnew_{\alpha}(\rho\|\sigma) \geq 0$. (We used that $X^{\dagger}X$ and $XX^{\dagger}$ have the same nonzero eigenvalues, where $X = \sigma^{\frac{1-\alpha}{2\alpha}} \rho^{\demi}$.)
If $\alpha < 1$, we directly have
$(\rho^{\frac12} \sigma^{\frac{1-\alpha}{\alpha}} \rho^{\frac12} )^{\alpha} \leq \rho$ again from operator monotonicity of $t \mapsto t^{\alpha}$. The desired statement then follows by considering that the prefactor $\frac{1}{\alpha-1}$ is negative in this case.
An analogous argument applies when $\sigma \leq \rho$. 

Additivity~(V) follows since $f(\rho \otimes \tau) = f(\rho) \otimes f(\tau)$ for all functions $f$ with the property $f(ab) = f(a) f(b)$. 
More precisely, the above property allows us to write
  \begin{eqnarray*}
    \Dnew_{\alpha}(\rho \otimes \tau \| \sigma \otimes \omega) &=& \frac{1}{\alpha-1} \log \frac{ \Tr \Big[\big(\sigma^{\aaa} \rho \sigma^{\aaa} \big)^\alpha \otimes \big( \omega^{\aaa} \tau \omega^{\aaa} \big)^\alpha \Big] }{\Tr[ \rho \otimes \tau ]} \\
    &=& \Dnew_{\alpha}(\rho\|\sigma) + \Dnew_{\alpha}(\tau\|\omega) \,.
  \end{eqnarray*}

Finally, (VI) follows since $f(\rho \oplus \tau) = f(\rho) \oplus f(\tau)$ and the trace term is thus additive. The property for $\Dnew_{\alpha}$ then follows by inspection and the choice $g_{\alpha}: t \mapsto \exp((\alpha-1)t)$.
\end{proof}

We need the following result in order to prove Theorem~\ref{thm:positive}. Let $\mathcal{E}_{\sigma}$ be a pinching in the eigenbasis of $\sigma$, i.e.\ the map $\rho \mapsto \sum_k |\psi_k\rangle\!\langle \psi_k| \rho |\psi_k\rangle\!\langle \psi_k|$ where $\{ \ket{\psi_k} \}_k$ are the eigenvectors of $\sigma$. Clearly, $\mathcal{E}_{\sigma}$ is a CPTPM.

\begin{prop} \label{prop:pinching}
  Let $\rho, \sigma \geq 0$ with $\rho \neq 0$ and $\alpha \in (0, 1) \cup (1, \infty)$. Then, we have
  \begin{align*}
    \Dnew_{\alpha}(\rho\|\sigma) \geq \Dnew_{\alpha}(\mathcal{E}_{\sigma}(\rho) \| \sigma) \,.
  \end{align*}
\end{prop}

\begin{proof}
  It suffices to show the claim for normalized $\rho$ and $\sigma$ since the pinching is a CPTPM and thus does not effect the trace.
  We have $\sigma^{\aaa} \mathcal{E}_{\sigma}(\rho) \sigma^{\aaa} = \mathcal{E}_{\sigma}\big(\sigma^{\aaa} \rho \sigma^{\aaa} \big)$ since the projectors $\ket{\psi_k}\!\bra{\psi_k}$ commute with $\sigma^{\aaa}$. For $\alpha > 1$, the case $\sigma \not\gg \rho$ is trivial, and otherwise we may write
  \begin{align*}
    \Dnew_{\alpha}(\mathcal{E}_{\sigma}(\rho) \| \sigma) = \frac{\alpha}{\alpha -1} \log \Big\| \mathcal{E}_{\sigma}\big(\sigma^{\aaa} \rho \sigma^{\aaa} \big) \Big\|_{\alpha} \leq 
    \frac{\alpha}{\alpha -1} \log \Big\| \sigma^{\aaa} \rho \sigma^{\aaa} \Big\|_{\alpha} = \Dnew_{\alpha}(\rho\|\sigma)
  \end{align*}
  where the inequality follows from the pinching inequality~\cite[Eq.~(IV.52)]{bhatia97} for the unitarily invariant Schatten $\alpha$ norm. For $\alpha < 1$, we use~\cite[Thm.~V.2.1]{bhatia97} which implies that $f_{\alpha}\big(\mathcal{E}_{\sigma}\big(\sigma^{\aaa} \rho \sigma^{\aaa} \big) \big) \geq \mathcal{E}_{\sigma} \big( f_{\alpha}\big(\sigma^{\aaa} \rho \sigma^{\aaa} \big) \big)$ for the operator concave function $f_{\alpha}: t \mapsto t^{\alpha}$~\cite[Thm.~V.19]{bhatia97}. Thus,
  \begin{align*}
    \Dnew_{\alpha}(\mathcal{E}_{\sigma}(\rho) \| \sigma) &= \frac{1}{\alpha -1} \log \tr \Big[ f_{\alpha}\Big( \mathcal{E}_{\sigma} \big( \sigma^{\aaa} \rho \sigma^{\aaa} \big) \Big)\Big]  \\
    &\leq \frac{1}{\alpha -1} \log \tr \Big[ \mathcal{E}_{\sigma} \Big( f_{\alpha}\big(\sigma^{\aaa} \rho \sigma^{\aaa} \big) \Big)\Big] = \Dnew_{\alpha}(\rho\|\sigma) \,. \qedhere
  \end{align*}
\end{proof}

\begin{proof}[Proof of Theorem~\ref{thm:positive}]
  Again, we separate the contributions due to the trace which leads to an additive term $\log \frac{\tr[\rho]}{\tr[\sigma]}$ which is positive when $\tr[\rho] \geq \tr[\sigma]$. Thus, it remains to show positivity for normalized $\rho$ and $\sigma$.
  By Proposition~\ref{prop:pinching} we can establish that
  $\Dnew_{\alpha}(\rho\|\sigma) \geq \Dnew_{\alpha}(\mathcal{E}_{\sigma}(\rho) \| \sigma) \geq 0$, inheriting this property from the commutative R\'enyi divergence~\cite{csiszar95}.
  
  If $\rho = \sigma$, we further have that $\Dnew_{\alpha}(\rho\|\sigma) = 0$ from Property~(IV).
\end{proof}

\begin{proof}[Proof of Proposition~\ref{thm:bigger-sigma}]
    First, assume that $\alpha > 1$ and that $\sigma \gg \rho$. Then, we express $\Dnew_{\alpha}$ via Definition \ref{def:ueber-rel-entropy} and write:
    \begin{align*}
        \Dnew_{\alpha}(\rho \| \sigma) &= \sup_{\tr[\tau]\leq 1} \Dnew_{\alpha}(\rho \| \sigma ; \tau)\\
        &= \frac{\alpha}{\alpha-1} \log \sup_{\tr[\tau] \leq 1 \atop \tau \gg \rho^{\demi} \sigma^{\frac{1-\alpha}{\alpha}} \rho^{\demi}} \tr\left[ \rho^{\onehalf} \sigma^{\frac{1-\alpha}{\alpha}} \rho^{\onehalf} \tau^{\frac{\alpha-1}{\alpha}}\right].
    \end{align*}
    Now, the function $x \mapsto x^{\frac{1-\alpha}{\alpha}}$ is operator monotone decreasing, since $\frac{1-\alpha}{\alpha} \in (-1,0)$. Hence, for any fixed $\tau$, we have that
    \[ \tr\left[ \rho^{\onehalf} \sigma^{\frac{1-\alpha}{\alpha}} \rho^{\onehalf} \tau^{\frac{\alpha-1}{\alpha}}\right] \geq \tr\left[ \rho^{\onehalf} {\sigma'}^{\frac{1-\alpha}{\alpha}} \rho^{\onehalf} \tau^{\frac{\alpha-1}{\alpha}}\right], \]
    and the result follows for this case by choosing the $\tau$ that achieves the maximum for $\sigma'$. If $\sigma \not\gg \rho$, then $\Dnew_{\alpha}(\rho \| \sigma) = \infty$ and the statement is trivially true.

    If $\alpha \in [\onehalf, 1)$, we simply write:
    \begin{align*}
        \Dnew_{\alpha}(\rho \| \sigma) = \frac{1}{\alpha-1} \log \tr\left[ (\rho^{\onehalf} \sigma^{\frac{1-\alpha}{\alpha}} \rho^{\onehalf})^{\alpha}\right].
    \end{align*}
    Now, since $x \mapsto x^{\frac{1-\alpha}{\alpha}}$ and $x \mapsto x^{\alpha}$ are both operator monotone, we have that
    \begin{align*}
        (\rho^{\onehalf} \sigma^{\frac{1-\alpha}{\alpha}} \rho^{\onehalf})^{\alpha} &\leq (\rho^{\onehalf} {\sigma'}^{\frac{1-\alpha}{\alpha}} \rho^{\onehalf})^{\alpha},
    \end{align*}
    and the result follows, noting that the prefactor $\frac{1}{\alpha-1}$ is negative.
\end{proof}

\subsection{Limits for $\alpha \to 1$ and $\alpha \to \infty$}
\label{sec:proof-limits}

In order to prove convergence to the von Neumann entropy, we first need to evaluate the following derivative.
The technique is taken from~\cite[Thm.~2.7]{boguslaw}.

\begin{prop}
  \label{smalldprop}
  Let $X, Y > 0$. Define $Z_{\alpha} = \yb X \yb$. Then, $(0,\infty) \ni \alpha \longmapsto \Tr[Z_{\alpha}^{\alpha}]$ is differentiable at $\alpha = 1$ and
  \begin{equation} \label{eq:derivativeatone}
    \frac{d}{d\alpha} \tr[Z_{\alpha}^{\alpha}] = \tr[Z_{\alpha}^{\alpha} \ln Z_{\alpha}] - \frac{1}{\alpha} \tr[Z_{\alpha}^{\alpha} \log Y] , \qquad
    \frac{d}{d\alpha}\bigg|_{\alpha=1} \tr[Z_{\alpha}^{\alpha}] = \ln 2 \; D(X\|Y).
  \end{equation}
\end{prop}

\begin{proof}
  For $\alpha \in (0,\infty)$ recall $Z_\alpha := \yb X \yb > 0$. We write
  \begin{eqnarray*}
    Z_{\alpha + h}^{\alpha + h} - Z_\alpha^\alpha & = & \int_0^1 ds \frac{d}{ds} Z_{\alpha + h}^{s(\alpha +h)} Z_\alpha^{(1-s)\alpha}\\
    & = & \int_0^1 ds \left( Z_{\alpha + h}^{s(\alpha +h)} \left( \ln{Z_{\alpha + h}^{\alpha + h}} - \ln Z^{\alpha}_\alpha \right) Z_\alpha^{(1-s)\alpha} \right).
  \end{eqnarray*}
  Taking the trace, we obtain
  \begin{eqnarray*}
    \Tr Z^{\alpha + h}_{\alpha + h} - \Tr Z^\alpha_\alpha & = & (\alpha + h) \int_0^1 ds \Tr[ Z_\alpha^{(1-s)\alpha} Z_{\alpha + h}^{s(\alpha + h)} (\ln Z_{\alpha + h} - \ln Z_\alpha)] \\
    && + h \int_0^1 ds \Tr[Z_\alpha^{(1-s)\alpha} Z_{\alpha + h}^{s(\alpha + h)} \ln Z_{\alpha}].
  \end{eqnarray*}
  We take the limit
  \begin{multline}\label{eq:plugderivative} 
    \lim_{h \rightarrow 0} \frac{1}{h} \left( \Tr Z^{\alpha + h}_{\alpha + h} - \Tr Z^\alpha_\alpha \right)\\
    \begin{split}
    &= \alpha \int_0^1 ds \Tr\left[ Z_\alpha^{(1-s)\alpha} Z_{\alpha}^{s\alpha} \lim_{h \rightarrow 0}\frac{1}{h} (\ln Z_{\alpha + h} - \ln Z_\alpha)\right] + \int_0^1 ds \Tr\left[Z_\alpha^{(1-s)\alpha} Z_{\alpha}^{s\alpha} \ln Z_{\alpha}\right]\\
    &= \alpha \Tr\left[Z_\alpha^\alpha \frac{d}{d\beta}\bigg|_{\alpha} \ln Z_\beta\right] + \Tr[Z_\alpha^\alpha \ln Z_\alpha].
    \end{split}
  \end{multline}
  Here, we have used the fact that $Z_\alpha$ is invertible (otherwise the product $Z_\alpha \ln Z_{\alpha + h}$ would not be well-defined).  The formula $\ln x = \int_0^\infty ds (\frac{1}{1+s}-\frac{1}{x+s})$ yields the integral representation
  \begin{equation*}
    \ln Z_\alpha = \int_0^\infty ds \left( \frac{1}{\ident + s} - \frac{1}{Z_\alpha + s} \right).
  \end{equation*}
  We use it to compute
  \begin{equation*}
    \frac{d}{d\alpha} \ln Z_\alpha = \int_0^\infty ds \left(\frac{1}{Z_\alpha + s}\right) \frac{dZ_\alpha}{d\alpha} \left(\frac{1}{Z_\alpha + s}\right).
  \end{equation*}
  Plugging this into Equation \eqref{eq:plugderivative} and using the cyclicity of the trace, we obtain
  \begin{equation*}
    \lim_{h \rightarrow 0} \frac{1}{h} \left( \Tr Z^{\alpha + h}_{\alpha + h} - \Tr Z^\alpha_\alpha \right) = \alpha \Tr\left[Z_\alpha^\alpha \int_0^\infty ds \left(\frac{1}{Z_\alpha + s}\right)^2 \frac{dZ_\alpha}{d\alpha}\right] + \Tr[Z_\alpha^\alpha \ln Z_\alpha].
  \end{equation*}
  Since $x^{-1} = \int_0^\infty ds \left(\frac{1}{s+x}\right)^2$, we obtain,
  \begin{equation} \label{eq:plugderivativeagain}
      \lim_{h \rightarrow 0} \frac{1}{h} \left( \Tr Z^{\alpha + h}_{\alpha + h} - \Tr Z^\alpha_\alpha \right) = \alpha \Tr\left[Z_\alpha^{\alpha-1} \frac{dZ_\alpha}{d\alpha}\right] + \Tr[Z_\alpha^\alpha \ln Z_\alpha].
  \end{equation}
  Next, we compute 
  \begin{eqnarray*}
    \frac{dZ_\alpha}{d\alpha} &=& \frac{d}{d\alpha}\; \yb X \yb \\
    & = & -\frac{1}{2\alpha^2} (\ln Y\; Z_\alpha + Z_\alpha \ln Y).
  \end{eqnarray*}
  One can verify this calculation using a spectral decomposition of $Y$. We plug this into Equation \eqref{eq:plugderivativeagain} and once again use the cyclicity of the trace to obtain
  \begin{equation*}
    \frac{d}{d\alpha} \tr[Z_{\alpha}^{\alpha}] = \lim_{h \rightarrow 0} \frac{1}{h} \left( \Tr Z^{\alpha + h}_{\alpha + h} - \Tr Z^\alpha_\alpha \right) = \Tr[Z_\alpha^\alpha \ln Z_\alpha] - \frac{1}{\alpha} \Tr[Z_\alpha^\alpha \ln Y ].
  \end{equation*}
  Taking the limit $\alpha \longrightarrow 1$ proves the lemma.
\end{proof}

We now turn to the proof of Theorem~\ref{thm:limits}.

\begin{proof}[Proof of Theorem~\ref{thm:limits}]
Note first that due to additivity and the normalization condition, we can always write $\Dnew_{\alpha}( \rho \| \sigma ) = \Dnew_{\alpha}( \tr[\rho] X \| \tr[\sigma] Y ) = \Dnew_{\alpha}(X\|Y) + \log \frac{\tr[\rho]}{\tr[\sigma]}$. Thus, it suffices to show convergence for normalized $X$ and $Y$. Recall that 
$$\Dnew_{\alpha}(X\|Y) = \frac{1}{\alpha -1} \log \Tr\left[ \Big(\yb X \yb\Big)^\alpha \right]$$ 
for all $\alpha \in (0,\infty)\setminus \{1\}$.

We first show convergence to the von Neumann entropy. First, assume that $X, Y > 0$. Then, by L'Hôpital's rule:
  \begin{eqnarray*}
    \lim_{\alpha \rightarrow 1} \Dnew_{\alpha}(X\|Y) &=& \lim_{\alpha \rightarrow 1} \frac{1}{\alpha -1} \log \Tr\left[ \left(\yb X \yb \right)^\alpha \right] \\
    &=& \frac{d}{d\alpha}\bigg|_{\alpha = 1} \log \Tr\left[ \left(\yb X \yb\right)^\alpha \right] \\
    &=& \frac{1}{\ln 2 \cdot \Tr[X]} \frac{d}{d\alpha}\bigg|_{\alpha = 1} \Tr\left[ \left(\yb X \yb\right)^\alpha \right] .
  \end{eqnarray*}
  We now apply Proposition~\ref{smalldprop}. This proposition also holds for $X,Y$ noninvertible with $Y \gg X$ by restricting the Hilbert space to $\supp Y$ (see also~\cite[p.~256]{boguslaw}). If $Y \not\gg X$, then $\Dnew_1(X\|Y) = \infty$ and likewise $\Dnew_{\alpha}(X\|Y) = \infty$ for $\alpha > 1$. It is also possible to show that $\Dnew_{\alpha}(X\|Y)$ then diverges as $\alpha \nearrow 1$. In this case, we note that $\lim_{\alpha \nearrow 1} \tr[(\yb X \yb)^{\alpha}] = \tr[PXP] < 1$ where $P$ is the projector onto the support of $Y$. Hence, $\lim_{\alpha \nearrow 1} \frac{1}{\alpha-1} \log \tr[(\yb X \yb)^{\alpha}] = \infty$.

To show convergence to the max relative entropy, we first write
  \begin{equation}\label{eq:lemmasplit} 
      \left\| \yb X \yb \right\|_{\alpha} = \left\| Y^{-\onehalf} X Y^{-\onehalf} \right\|_{\alpha} + \left\| \yb X \yb \right\|_{\alpha} - \left\| Y^{-\onehalf} X Y^{-\onehalf} \right\|_{\alpha}.
  \end{equation}
  By the reverse triangle inequality for the $\alpha$ norm on $\C^n$, we have that
  \begin{equation*}
      \left| \| \yb X \yb \|_{\alpha} - \| Y^{-\onehalf} X Y^{-\onehalf} \|_{\alpha} \right| \leq \left\| \yb X \yb - Y^{-\onehalf} X Y^{-\onehalf} \right\|_{\alpha}.
  \end{equation*}
  So we have that
  \begin{align*}
      \lim_{\alpha \rightarrow \infty} \Dnew_{\alpha}(X\|Y) &= \lim_{\alpha \rightarrow \infty} \frac{\alpha}{\alpha-1} \log \left\| \yb X \yb \right\|_{\alpha}\\
      &\leq  \log \left( \lim_{\alpha \rightarrow \infty}\left\| Y^{-\onehalf} X Y^{-\onehalf} \right\|_{\alpha} + \lim_{\alpha \rightarrow \infty}\left\| \yb X \yb - Y^{-\onehalf} X Y^{-\onehalf} \right\|_{\alpha} \right)\\
      &\leq \log \left( \lim_{\alpha \rightarrow \infty} \left\| Y^{-\onehalf} X Y^{-\onehalf} \right\|_{\alpha} + (\dim \mathcal{H})\lim_{\alpha \rightarrow \infty}  \left\| \yb X \yb - Y^{-\onehalf} X Y^{-\onehalf} \right\|_{\infty} \right)\\
      &= \log \left\| Y^{-\onehalf} X Y^{-\onehalf} \right\|_{\infty}\\
      &= D_{\max}(X\| Y).
  \end{align*}
  Likewise,
  \begin{align*}
      \lim_{\alpha \rightarrow \infty} \Dnew_{\alpha}(X\|Y) &= \lim_{\alpha \rightarrow \infty} \frac{\alpha}{\alpha-1} \log \left\| \yb X \yb \right\|_{\alpha}\\
      &\geq  \log \left( \lim_{\alpha \rightarrow \infty}\left\| Y^{-\onehalf} X Y^{-\onehalf} \right\|_{\alpha} - \lim_{\alpha \rightarrow \infty}\left\| \yb X \yb - Y^{-\onehalf} X Y^{-\onehalf} \right\|_{\alpha} \right)\\
      &\geq \log \left( \lim_{\alpha \rightarrow \infty} \left\| Y^{-\onehalf} X Y^{-\onehalf} \right\|_{\alpha} - (\dim \mathcal{H})\lim_{\alpha \rightarrow \infty}  \left\| \yb X \yb - Y^{-\onehalf} X Y^{-\onehalf} \right\|_{\infty} \right)\\
      &= \log \left\| Y^{-\onehalf} X Y^{-\onehalf} \right\|_{\infty}\\
      &= D_{\max}(X\| Y). \qedhere
  \end{align*} 
\end{proof}

\subsection{Joint Convexity and Data-Processing}
\label{sec:proof-dp}

In order to prove Theorem~\ref{thm:data-proc}, it is sufficient to prove that  {$\exp((\alpha-1)\Dnew_{\alpha}(\cdot\|\cdot))$} is jointly convex. We break the proof up into three small lemmas, where $\mathcal{P}$ denotes the set of positive semi-definite operators and $\mathcal{P}_+$ the set of strictly positive operators.
\begin{lem}
  \label{lem:rotimesr}
  For $\beta \in [0,1]$, the map $F: \mathcal{P} \oplus \mathcal{P} \ni (L,R) \longmapsto L^\beta \otimes (R^{\T})^{1-\beta}$ is jointly operator concave.
\end{lem}

\begin{proof}
    Applying Theorem 5.14 from \cite{wolf-ln} to $h: L \longmapsto L \otimes \ident, g: R\longmapsto \ident \otimes R^{\T}$, and $f(x) = -x^\beta$ shows that the map $F$ restricted to $\mathcal{P} \oplus \mathcal{P}_+$ is jointly operator concave.
  
  Now, let $\lambda \in (0,1)$ and $ L_1,L_2,R_1,R_2 \in \mathcal{P}$. Since invertible matrices are dense in the space of matrices, there exist families $\{R_1^{(n)}\}_n$ and $\{R_2^{(n)}\}_n$ of operators in $\mathcal{P}_+$ that converge to $R_1$ and $R_2$, respectively. Moreover, $\lambda R_1^{(n)} + (1-\lambda) R_2^{(n)}$ is strictly positive for all $n$.  Hence,
  \begin{equation*}
    F\left(\lambda L_1 + (1-\lambda) L_2, \lambda R_1^{(n)} + (1-\lambda) R_2^{(n)}\right) \leq \lambda F\left(L_1, R_1^{(n)}\right) + (1-\lambda) F\left(L_2, R_2^{(n)}\right).
  \end{equation*}
  The claim follows from continuity of $F$ in the second argument in the limit $n\longrightarrow \infty$.
\end{proof}

\begin{lem}
  \label{lem:cancel}
  For $\alpha \in [1,2]$ and $\beta \in [0,1]$, the map 
  \[ F: \mathcal{P} \oplus \mathcal{P}_+ \ni (L,R) \longmapsto R^{\beta/2}\left( R^{-\beta/2}LR^{-\beta/2}\right)^\alpha R^{\beta/2} \otimes (R^{\T})^{(1-\alpha)(1-\beta)} \]
  is jointly operator convex.
\end{lem}

\begin{proof}
    By Lemma \ref{lem:rotimesr}, the map $g: \mathcal{P} \ni R \longmapsto R^\beta \otimes \left(R^{\T}\right)^{1-\beta}$ is operator concave. It is also positive. Moreover, $h: \mathcal{P} \ni L \longmapsto L \otimes \ident$ is positive and affine. Since $f(x) = x^\alpha$ is operator convex, Theorem 5.14 from \cite{wolf-ln} proves the claim.
\end{proof}

\begin{lem}
  \label{lem:alpha12}
  Let $\alpha \in [1,2]$. The functional $\exp((\alpha-1)\Dnew_{\alpha}(\cdot\|\cdot)$ is jointly convex.
\end{lem}

\begin{proof}
Let $\gamma$ be defined by $\ket{\gamma} = \sum_i |i\rangle \otimes |i \rangle$, where $\{|i\rangle\}$ is some orthonormal basis. Define $\beta := 1-1/\alpha \in [0,1/2]$ and note that $\beta + (1-\beta)(1-\alpha) = 0$. We use this to express 
  \begin{align*}
    &\Big\langle \gamma \Big| R^{\beta/2}\left( R^{-\beta/2}LR^{-\beta/2}\right)^\alpha R^{\beta/2} \otimes (R^{\T})^{(1-\alpha)(1-\beta)} \Big| \gamma \Big\rangle \\
    \begin{split}
    &\qquad = \Tr[R^{\beta/2}\left( R^{-\beta/2}LR^{-\beta/2}\right)^\alpha R^{\beta/2}R^{(1-\alpha)(1-\beta)}] \\
    &\qquad= \Tr[(R^{-\beta/2}LR^{-\beta/2})^\alpha R^{\beta+(1-\alpha)(1-\beta)}] \\
    &\qquad= \Tr[(R^{\frac{1}{2\alpha}-\frac{1}{2}} L R^{\frac{1}{2\alpha}-\frac{1}{2}})^\alpha]\\
    &\qquad= \exp((\alpha-1)\Dnew_{\alpha}(L\|R)).
    \end{split}
  \end{align*}
Now we apply Lemma \ref{lem:cancel} to this quantity to conclude the proof of the lemma.
\end{proof}

\subsection{Monotonicity in $\alpha$}
\label{sec:proof-mono}

We now show that $\Dnew_{\alpha}(\rho\|\sigma)$ is monotonous in $\alpha$ by first proving it for the auxiliary quantity $\Dnew_{\alpha}(\rho \| \sigma; \tau)$, introduced in Definition \ref{def:ueber-rel-entropy}.
\begin{lem} \label{lm:aux-mono}
  Let $\rho \geq 0$ with $\tr[\rho] = 1$, and let $\sigma, \tau \geq 0$. Then, $\alpha \mapsto \Dnew_{\alpha}(\rho\|\sigma; \tau)$ is monotonically increasing in $\alpha$. 
\end{lem}
\begin{proof}
  First assume $\sigma, \tau \gg \rho$.
  Let $\ket{\phi}$ be a purification of $\rho$. Furthermore, set $\beta = \frac{\alpha-1}{\alpha}$ and $X = \sigma^{-1} \otimes \tau^{T}$, where the transpose is taken with regards to the Schmidt bases of $\ket{\phi}$. Then,
  \begin{align*}
    \Dnew_{\alpha}(\rho\|\sigma;\tau) = \frac{1}{\beta} \log \tr\big( \rho^{\demi} \sigma^{-\beta} \rho^{\demi} \tau^{\beta} \big) =
    \frac{1}{\beta} \log \bra{\phi} X^{\beta} \ket{\phi} \,.
  \end{align*}
  Since $\frac{\mathrm{d}\beta}{\mathrm{d}\alpha} = \frac{1}{\alpha^2}$, we have
  \begin{align*}
    \frac{\mathrm{d}}{\mathrm{d}\alpha} \Dnew_{\alpha}(\rho\|\sigma;\tau) = \frac{1}{\alpha^2} \frac{\mathrm{d}}{\mathrm{d}\beta} \underbrace{\bigg( \frac{1}{\beta} \log \bra{\phi} X^{\beta} \ket{\phi} \bigg)}_{f(\beta)}
  \end{align*}
  Now we proceed similarly to~\cite[Lm.~3]{tomamichel08}, where monotonicity for $\Dold_{\alpha}$ is shown. We find
  \begin{align*}
    \frac{\mathrm{d}}{\mathrm{d}\beta} f(\beta) &= -\frac{1}{\beta^2} \log \bra{\phi}X^{\beta}\ket{\phi} + \frac{1}{\beta \bra{\phi}X^{\beta}\ket{\phi}} \bra{\phi}X^{\beta} \log X\ket{\phi} \\
    &= \frac{\bra{\phi}g(X^{\beta})\ket{\phi} - g(\bra{\phi}X^{\beta}\ket{\phi})}{\beta^2 \bra{\phi}X^{\beta}\ket{\phi}}
  \end{align*}
  Here, $g: x \mapsto x \log x$ is a convex function and Jensen's inequality thus implies that the derivative is non-negative. If $\sigma \not\gg \rho$, then $\Dnew_{\alpha}(\rho\|\sigma,\tau) = \infty$ for $\alpha > 1$. Similarly, if $\tau \not\gg \rho$, then $\Dnew_{\alpha}(\rho\|\sigma,\tau) = -\infty$ for $\alpha < 1$. The proposition thus holds trivially in these cases, concluding the proof.
\end{proof}

\begin{proof}[Proof of Theorem~\ref{thm:mono}]
  It suffices to show this property for normalized $\rho$ and $\sigma$. Then, for any $\alpha' \geq \alpha$ and for all $\tau \geq 0$, Lemma~\ref{lm:aux-mono} implies that $D_{\alpha'}(\rho\|\sigma; \tau) \geq \Dnew_{\alpha}(\rho\|\sigma; \tau)$. In particular, this holds true for the $\tau$ maximizing $\Dnew_{\alpha}(\rho\|\sigma)$ in~\eqref{eq:ueber-to-rel}, which concludes the proof.
\end{proof}

\subsection{Duality of the Conditional R\'enyi Entropy}
\label{sec:proof-dual}

\begin{proof}[Proof of Theorem \ref{thm:dual}]
    We write $\rho_{ABC} = \ket{\phi}\!\bra{\phi}$. Let $\ket{\phi} = \sum_i r_i \ket{i}_{AB} \otimes \ket{i}_C$ be a Schmidt decomposition for $\ket{\phi}$ and define $\ket{\psi} := \sum_i \ket{i}_{AB} \otimes \ket{i}_C$ as the (unnormalized) maximally entangled state in these bases.  It is easy to verify that
    \begin{align*}
      \tr\big( \rho_{AB}^{\demi} C_{AB} \rho_{AB}^{\demi} D_{AB} \big) = \bra{\psi} \rho_{AB}^{\demi} C_{AB} \rho_{AB}^{\demi} \otimes D_{C} \ket{\psi} = \bra{\phi} C_{AB} \otimes D_C \ket{\phi} ,
    \end{align*}
    where $D_C$ is the transpose of $D_{AB}$ with regards to the bases $\{\ket{i}_{AB}\}$ and $\{\ket{i}_C\}$, i.e. $D_{AB} \otimes \id_C \ket{\psi} = \id_{AB} \otimes D_C \ket{\psi}$.

 In the following, the suprema and infima are taken over operators $\sigma_B, \tau_C \geq 0$ with $\tr[\sigma_B] \leq 1$ and $\tr[\tau_C] \leq 1$.     Thus, using Lemma~\ref{lem:extended-norm-duality}, we then derive the following convenient representation of the conditional R\'enyi entropy as a minimax problem.
     \begin{align}
        \Hnew_{\alpha}(A|B)_{\rho} &= 
        \begin{cases}
        \frac{\alpha}{1-\alpha} \log \sup_{\sigma_B}  \Big\| \rho_{AB}^{\demi} \id_A \otimes\, \sigma_B^{\frac{1}{\alpha}-1} \rho_{AB}^{\demi} \Big\|_{\alpha} & \text{if $\alpha<1$}\\
        \frac{\alpha}{1-\alpha} \log \inf_{\sigma_B}  \Big\| \rho_{AB}^{\demi} \id_A \otimes\, \sigma_B^{\frac{1}{\alpha}-1} \rho_{AB}^{\demi} \Big\|_{\alpha} & \text{if $\alpha>1$} 
    \end{cases}\  \nonumber\\
        &= 
        \begin{cases}
        \frac{\alpha}{1-\alpha} \log \sup_{\sigma_B} \inf_{\tau_C} \Big\langle \phi \Big| \id_A \otimes\, \sigma_B^{\frac{1}{\alpha}-1} \otimes \tau_C^{1-\frac{1}{\alpha}} \Big|\phi\Big\rangle & \text{if $\alpha<1$}\\
        \frac{\alpha}{1-\alpha} \log \inf_{\sigma_B} \sup_{\tau_C} \Big\langle \phi \Big| \id_A \otimes\, \sigma_B^{\frac{1}{\alpha}-1} \otimes \tau_C^{1-\frac{1}{\alpha}} \Big|\phi\Big\rangle & \text{if $\alpha>1$} 
    \end{cases}\  \label{eq:symmetric-bc} .
    \end{align}
 Furthermore, note that $\frac{\alpha}{1-\alpha} = - \frac{\beta}{1-\beta}$ so that \eqref{eq:symmetric-bc} is concave in $\sigma_B$ and convex in $\tau_C$ if $\alpha < 1$ and the other way around if $\alpha > 1$. This means that the infimum and supremum can be interchanged in both cases~\cite{sion58}. Hence, written in this form, the expressions for $\Hnew_{\alpha}(A|B)_{\rho}$ and $-\Hnew_{\beta}(A|C)_{\rho}$ coincide, which establishes the claim.
\end{proof}

\subsection{Conditioning on Classical Information}
\label{sec:proof-class}

\begin{proof}[Proof of Proposition~\ref{pr:class}]
We consider a normalized tripartite state $\rho_{ABY}$ with a classical $Y$, i.e., $\rho_{ABY} = \bigoplus_y p_y \rho_{AB}^y$. 
Recall that by definition, 
$$
\Hnew_{\alpha}(A|BY)_{\rho} = -\inf_{\sigma_{BY}} \Dnew_{\alpha}(\rho_{ABY}\|\id_A \otimes\, \sigma_{BY})
$$
where the infimum is over all (normalized) states $\sigma_{BY}$, but due to data processing (we can measure the $Y$-register, which does not affect $\rho_{ABY}$), we can restrict to states $\sigma_{BY}$ with classical~$Y$. Using the decomposition of $\Dnew_{\alpha}$ into the divergences of the conditional states, we then obtain
\begin{align*}
\Hnew_{\alpha}(A|BY)_{\rho} &= -\inf_{\{q_y\},\{\sigma_{B}^y\}} \frac{1}{\alpha-1} \log \sum_y p_y^\alpha q_y^{1-\alpha} \exp\Bigl((\alpha-1) \Dnew_{\alpha}(\rho_{AB}^y\|\id_A \otimes \,\sigma_B^y)\Bigr) \\
&= \inf_{\{q_y\}} \frac{1}{1-\alpha}\log \sum_y p_y^\alpha q_y^{1-\alpha} \exp\Bigl((1-\alpha) \Hnew_{\alpha}(A|B)_{\rho^y}\Bigr)  \, .
\end{align*}
Writing $r_y = p_y \exp\bigl(\frac{1-\alpha}{\alpha} \Hnew_{\alpha}(A|B)_{\rho^y}\bigr)$, and using straightforward Lagrange multiplier technique, one can show that the infimum is attained by $\hat q_y = r_y/\sum_z r_z$. Thus, 
\begin{align*}
\Hnew_{\alpha}(A|BY)_{\rho} &= \frac{1}{1-\alpha}\log \sum_y r_y^\alpha \hat q_y^{1-\alpha} 
= \frac{1}{1-\alpha}\log\sum_y r_y^\alpha \biggl(\frac{\sum_z r_z}{r_y}\biggr)^{\alpha-1}  \\
&= \frac{\alpha}{1-\alpha}\log \sum_y r_y = \frac{\alpha}{1-\alpha}\log \sum_y p_y \exp\Bigl(\textstyle\frac{1-\alpha}{\alpha} \Hnew_{\alpha}(A|B)_{\rho^y}\Bigr)  \, . \qedhere
\end{align*}
\end{proof}

\paragraph*{Acknowledgments.}

We thank M.~Wilde for comments on an early draft.
FD acknowledges support from the Danish National Research Foundation and The National Science Foundation of China (under the grant 61061130540) for the Sino-Danish Center for the Theory of Interactive Computation, within which part of this work was performed; and also from the CFEM research center (supported by the Danish Strategic Research Council) within which part of this work was performed. OS acknowledges financial support from the Elite Network of Bavaria, project QCCC. MT is funded by the Ministry of Education (MOE) and National Research Foundation Singapore, as well as MOE Tier 3 Grant ``Random numbers from quantum processes'' (MOE2012-T3-1-009).

\end{document}